\newif\ifcomments 
\newif\ifauthors  

\commentsfalse
\authorstrue

\documentclass[11pt,pdfa]{article}%
\pdfoutput=1 

\usepackage{iftex}
\ifPDFTeX
  \usepackage[utf8]{inputenc}
  \usepackage[noTeX]{mmap}
  \usepackage[T1]{fontenc}
\fi
\ifLuaTeX
  \usepackage{luatex85}
  \usepackage[noTeX]{mmap}
\fi

\usepackage{amsmath,amsfonts,amssymb}
\usepackage{amsthm}
\usepackage{mathtools}
\usepackage{graphicx}
\usepackage[margin=1in]{geometry} 
\usepackage[pagebackref,linktocpage]{hyperref}
\hypersetup{
    colorlinks=true, 
    linkcolor=blue, 
    citecolor=blue, 
    urlcolor=blue 
}
\usepackage{microtype}
\usepackage{thm-restate}
\usepackage[capitalize,nameinlink]{cleveref}

\usepackage{tikz}
\usepackage{pgfplots}
\pgfplotsset{compat=1.13}
\usepackage{wrapfig}
\usepackage{enumerate}
\usepackage{braket}
\usepackage{qcircuit}
\usepackage{circledsteps} 
\usepackage{dsfont} 

\usepackage{thm-restate}

\renewcommand{\backref}[1]{}

\renewcommand{\backrefalt}[4]{%
\ifcase #1 %
\or
[p.\ #2]%
\else
[pp.\ #2]%
\fi}

\newtheorem{theorem}{Theorem}

\newtheorem{claim}[theorem]{Claim}

\newtheorem{conjecture}[theorem]{Conjecture}
\newtheorem{corollary}[theorem]{Corollary}

\newtheorem{definition}[theorem]{Definition}

\newtheorem{fact}[theorem]{Fact}
\newtheorem{lemma}[theorem]{Lemma}

\newtheorem{proposition}[theorem]{Proposition}

\newtheorem{property}[theorem]{Property}

\Crefformat{enumi}{#2Item~(#1)#3}
\crefformat{enumi}{Item~(#2#1#3)}

\ifcomments
  \newcommand{\comment}[1]{{\color{red}#1}}
\else
  \newcommand{\comment}[1]{}
\fi

\usepackage{mleftright}
\newcommand{\mparen}[1]{\mleft(#1\mright)}
\newcommand{\mbracket}[1]{\mleft[#1\mright]}

\newcommand{\abs}[1]{\left|#1\right|}

\newcommand{\eps}{\varepsilon}
\newcommand{\sgn}{\mathrm{sgn}}

\newcommand{\equad}{\mathrel{\phantom{=}}}
\newcommand{\E}{\mathop{\mathbf{E}}}
\newcommand{\Naturals}{\mathbb{N}}
\newcommand{\Reals}{\mathbb{R}}
\newcommand{\poly}{\mathrm{poly}}
\newcommand{\negl}{\mathrm{negl}}
\newcommand{\polylog}{\mathrm{polylog}}
\newcommand{\quasipoly}{\mathrm{quasipoly}}
\newcommand{\AND}{\normalfont\textsc{AND}} 
\newcommand{\OR}{\normalfont\textsc{OR}}
\newcommand{\NOT}{\normalfont\textsc{NOT}}

\newcommand{\trnc}{\mathrm{trnc}}
\newcommand{\diag}{\mathrm{diag}}
\newcommand{\TD}{\mathrm{TD}}
\newcommand{\TVD}{\mathrm{TVD}}
\newcommand{\ind}{\mathop{\mathds{1}}}
\newcommand{\ttable}{\mathrm{tt}}

\newcommand{\var}{{\mathop{\bf Var\/}}}
\newcommand{\cov}{{\mathop{\bf Cov\/}}}

\begin{document}

\title{Quantum Cryptography in Algorithmica}

\newcommand{\email}[1]{\href{mailto:#1}{\texttt{#1}}}

\ifauthors
\author{William Kretschmer\thanks{University of Texas at Austin. \ Email:
\email{kretsch@cs.utexas.edu}. \ Supported by an NDSEG Fellowship.} \and 
Luowen Qian\thanks{Boston University. \ Email: \email{luowenq@bu.edu}. \ Supported by DARPA under Agreement No.\ HR00112020023.}
\and
Makrand Sinha\thanks{Simons Institute and University of California at Berkeley. \ Email:
\email{makrand@berkeley.edu}. \ Supported by a Simons-Berkeley Postdoctoral Fellowship.}
\and
Avishay Tal\thanks{University of California at Berkeley. \ Email: \email{atal@berkeley.edu}. \ Supported by a Sloan Research Fellowship and NSF CAREER Award CCF-2145474.}
}
\fi

\date{}
\maketitle

\begin{abstract}
    We construct a classical oracle relative to which $\mathsf{P} = \mathsf{NP}$ yet single-copy secure pseudorandom quantum states exist.
    In the language of Impagliazzo's five worlds, this is a construction of pseudorandom states in ``Algorithmica,''  and hence shows that in a black-box setting, quantum cryptography based on pseudorandom states is possible even if one-way functions do not exist.
    \comment{W: Remove this? Our result can equivalently be understood as showing that in the random oracle model, there exist pseudorandom quantum states that are secure against $\mathsf{BQP}^\mathsf{PH}$ adversaries.}
    As a consequence, we demonstrate that there exists a property of a cryptographic hash function that simultaneously (1) suffices to construct pseudorandom states, (2) holds for a random oracle, and (3) is independent of $\mathsf{P}$ vs.\ $\mathsf{NP}$ in the black-box setting.
    \comment{W: Remove this? These pseudorandom states are based on interleaving Hadamard gates with random diagonal unitaries.}
    We also introduce a conjecture that would generalize our results to multi-copy secure pseudorandom states.
    
    We build on the recent construction by Aaronson, Ingram, and Kretschmer (CCC 2022) of an oracle relative to which $\mathsf{P} = \mathsf{NP}$ but $\mathsf{BQP} \neq \mathsf{QCMA}$, based on hardness of the $\OR \circ \textsc{Forrelation}$ problem.
    Our proof also introduces a new discretely-defined variant of the Forrelation distribution, for which we prove pseudorandomness against $\mathsf{AC^0}$ circuits. This variant may be of independent interest.
\end{abstract}

\section{Introduction}

One-way functions (OWFs) have played a central role in computational cryptography since its birth \cite{DH76-crypto}. 
On the one hand, the existence of one-way functions would separate $\mathsf P$ from $\mathsf{NP}$ in an average case sense; on the other, their existence has proven to be necessary for almost all classical cryptographic tasks \cite{IL89-owf,Gol90-prg}. This reveals a fundamental tension in the classical world: we cannot expect to solve all problems in $\mathsf{NP}$ extremely well but also have useful cryptography at the same time.
In the language of Impagliazzo's five worlds \cite{Imp95-average}, there is no hope of constructing much useful cryptography in ``Algorithmica,'' a world in which $\mathsf{P} = \mathsf{NP}$. Rather, the ``Minicrypt'' world, where one-way functions exist, is generally considered to capture the bare minimum of cryptography, because one-way functions are implied by almost all other cryptographic primitives. At the same time, even just assuming the existence of one-way functions, a wide range of other cryptographic primitives are possible, including pseudorandom generators, pseudorandom functions, symmetric-key encryption schemes, and digital signatures.



A growing body of work has shown that $\mathsf P \neq \mathsf{NP}$ may not be necessary to construct various useful \textit{quantum} cryptography. Quantum key distribution (QKD) \cite{BB84-qkd} is arguably the earliest demonstration of this idea: it enables two parties to securely exchange a secret key, assuming only that they share an untrusted quantum channel and an authenticated classical channel. The security proof of QKD is information-theoretic, and relies on no computational assumptions \cite{Ren08-qkd}. By contrast, classical key exchange lies in Impagliazzo's ``Cryptomania'' world, meaning that it relies on computational assumptions that appear to be even stronger than the existence of one-way functions \cite{IR89-limits}.

Unfortunately, much like classically, many interesting cryptographic tasks still remain impossible for information-theoretically secure quantum protocols \cite{May97-commitment,LC97-commitment}, and thus require some assumptions on the model, e.g.\ a computational bound on the adversary.
More recent works have demonstrated the possibility of building computationally secure quantum cryptography based on computational assumptions that are plausibly weaker than the existence of one-way functions. A prominent example is the construction of cryptography based on \textit{pseudorandom quantum states} (PRSs), introduced by Ji, Liu, and Song \cite{JLS18-prs}.
Informally, an ensemble of quantum states is pseudorandom if the states can be efficiently generated, and if no polynomial-time quantum adversary can distinguish a random state drawn from the ensemble from a Haar-random state.
PRSs can be defined with either single- \cite{MY22-prs} or multi-copy security \cite{JLS18-prs}, depending on whether the adversary is allowed a single copy, or any polynomial number of copies of the unknown state, respectively.
\cite{JLS18-prs} also showed that the existence of quantum-secure one-way functions is sufficient to construct multi-copy pseudorandom states, although the converse is not known. Hence, assuming the existence of PRSs is no stronger than assuming the existence of quantum-secure OWFs.

Despite appearing weaker than one-way functions, pseudorandom states are surprisingly powerful, and suffice to construct a wide variety of cryptography.
Even with only single-copy secure pseudorandom states, we can already construct commitment schemes and some form of non-trivial one-time signatures \cite{AQY22-prs,MY22-prs}, the former of which are equivalent to secure multiparty computation and computational zero knowledge proofs \cite{YWLQ15-commitment,Yan22-commitments,BCKM21-multiparty,GLSV21-miniqcrypt,AQY22-prs,BCQ23-efi}.
From the more standard notion of multi-copy security, we can also achieve private-key query-secure quantum money \cite{JLS18-prs} and non-trivial one-time encryption \cite{AQY22-prs}.

Nevertheless, the evidence to suggest that PRSs are actually a weaker assumption than OWFs is extremely limited. Indeed, other than the basic intuition that there is no obvious way to construct OWFs out of PRSs, the only provable separation between these primitives is the result of Kretschmer \cite{Kre21-pseudorandom}, who constructed an oracle relative to which $\mathsf{BQP} = \mathsf{QMA}$ and yet pseudorandom states exist. This shows that, in the black box setting, quantum-secure OWFs are not implied by PRSs, because quantum algorithms can efficiently invert any classical function if $\mathsf{NP} \subseteq \mathsf{BQP}$.

However, Kretschmer's result comes with the major caveat that the oracle achieving this separation is quantum, meaning that the oracle is some arbitrary unitary transformation that only a quantum algorithm can query. Hence, perhaps it is unsurprising that this quantum oracle lets us achieve quantum cryptography (PRSs) but not classical cryptography (OWFs)---there is no meaningful way to define classical queries to a unitary oracle!
In other words, even though $\mathsf{BQP} = \mathsf{QMA}$ implies that quantum-secure OWFs do not exist, the statement ``$\mathsf{BQP}^\mathcal{O} = \mathsf{QMA}^\mathcal{O}$ implies that quantum-secure OWFs do not exist relative to $\mathcal{O}$'' is less meaningful when $\mathcal{O}$ is a quantum oracle, because any construction of OWFs cannot depend on $\mathcal{O}$.
Furthermore, quantum oracle separations are conceptually weaker than classical oracle separations, because they can produce consequences that fail relative to all classical oracles.
Indeed, Aaronson \cite{Aar09-qma1} observed that that there exist inclusions of complexity classes that trivially hold relative to all classical oracles (e.g.\ $\mathsf{BQP} \subseteq \mathsf{ZQEXP}$), but that can be separated relative to certain quantum oracles. Thus, for all we know, the result of \cite{Kre21-pseudorandom} could be merely an artifact of the quantumness of the oracle, and there could be a classical relativizing proof that PRSs imply OWFs!

Another conceptual limitation of Kretschmer's separation is that the pseudorandom state construction involves directly generating the state using a Haar-random oracle. However, unlike the classical random oracle model, we very much do not know how to even heuristically instantiate such a Haar random oracle in the real world, other than via ad-hoc approaches such as random quantum circuits. Therefore, \cite{Kre21-pseudorandom} offers little insight into plausible constructions of PRSs without OWFs in a non-oracular setting.

\subsection{Our Results}

In this work, we overcome these limitations of quantum oracles by constructing a separation of PRSs and OWFs relative to a \textit{classical} oracle. Our main result is the following:

\begin{theorem}[\Cref{prop:oracle_p_equals_np} and \Cref{thm:single_copy_prs_prob_1}, informal]
\label{thm:main_informal}
    There exists a classical oracle relative to which $\mathsf{P} = \mathsf{NP}$ and single-copy pseudorandom state ensembles exist.
\end{theorem}

\Cref{thm:main_informal} can thus be taken as a relativized construction of PRSs in Impagliazzo's ``Algorithmica'' \cite{Imp95-average}. Since OWFs do not exist if $\mathsf{P} = \mathsf{NP}$, our result shows that OWFs are not necessary to construct PRSs in the classical black box setting, answering a question of Ji, Liu, and Song \cite{JLS18-prs}. Note that our result is formally incomparable to the result of \cite{Kre21-pseudorandom}: on the one hand, our separation is conceptually stronger, because our oracle is classical, rather than quantum. On the other hand, we achieve single-copy PRSs, whereas \cite{Kre21-pseudorandom} achieves multi-copy PRSs.\footnote{Another technical difference is that we construct a world where $\mathsf{P} = \mathsf{NP}$, whereas \cite{Kre21-pseudorandom} constructs a world where $\mathsf{BQP} = \mathsf{QMA}$, and these are also incomparable \cite[Theorems 4 and 9]{AIK21-acrobatics}.} 

We briefly describe the oracle and the associated construction of pseudorandom states. Our oracle $\mathcal{O} = (A, B)$ consists of two parts: a random oracle $A$, and an oracle $B$ that is defined recursively to answer all possible $\mathsf{NP}$ predicates of either $A$ or $B$. Note that similar oracles were used in \cite{BM99-relativized, AIK21-acrobatics}, and that $\mathsf{P}^\mathcal{O} = \mathsf{NP}^\mathcal{O}$ essentially by definition. Furthermore, $B$ is constructed so that queries to $B$ are roughly equivalent in power to queries to $\mathsf{PH}^A$. So, our result can also be interpreted as showing that in the random oracle model, there exist PRSs that are secure against $\mathsf{BQP}^\mathsf{PH}$ adversaries.

Our pseudorandom state ensemble is defined using what we call \textit{$t$-Forrelation states}, which are $n$-qubit states $\ket{\Phi_F}$ of the form:
\[
\ket{\Phi_F} = U_{f^t} \cdot H \cdot U_{f^{t-1}} \cdot H \cdots H \cdot U_{f^1} \ket{+^n},
\]
where $F = (f^1, f^2, \ldots, f^t)$ is a $t$-tuple of $n$-bit Boolean functions $f^i: \{\pm 1\}^n \to \{\pm 1\}$,  $U_{f^i}$ is the phase oracle corresponding to $f^i$,
and $H$ is the $n$-qubit Hadamard transform. $t$-Forrelation states are a generalization of so-called ``phase states'', which correspond to the case $t = 1$ \cite{JLS18-prs, BS19-binary, INNRY22-qma-search, ABDY22-phase}.
$t$-Forrelation states are so called because of their connection to the $t$-fold $\textsc{Forrelation}$ problem \cite{AA18-forrelation, BS21}. 
We take the states in our PRS ensemble to be a set of randomly-chosen $2$-Forrelation states with $F$ specified by the random oracle. That is, we view $A$ as defining a pair of random functions $(f_k, g_k)$ for each key $k \in \{0,1\}^\kappa$ (where $\kappa$ is the security parameter), and we take the pseudorandom state keyed by $k$ to be:
\[
\ket{\varphi_k} \coloneqq \ket{\Phi_{(f_k, g_k)}} = U_{g_k} \cdot H \cdot U_{f_k} \ket{+^n}.
\]
Here, we can pick $n$ to be any polynomial in $\kappa$, although since we only achieve single-copy security, it is only non-trivial if $n > \kappa$, as also observed in prior works \cite{JLS18-prs,MY22-prs}.

We remark that similar ``Hadamard/phase cocktails''\footnote{We credit Scott Aaronson (personal communication) for suggesting this term.} have appeared elsewhere in quantum information \cite{AK07-qcma-qma,NHMW17-2design,NHKW17-pseudorandom}, including also in Ji, Liu, and Song's candidate construction of \textit{pseudorandom unitaries} \cite[Section 6.2]{JLS18-prs}, which are a strengthening of pseudorandom states.

\subsection{Proof Overview}

Our proof builds on the recent construction by Aaronson, Ingram, and Kretschmer \cite{AIK21-acrobatics} of an oracle relative to which $\mathsf{P} = \mathsf{NP}$ and $\mathsf{BQP} \neq \mathsf{QCMA}$. Their proof of this separation involves showing that an oracle distinguishing problem called $\OR \circ \textsc{Forrelation}$ is not in $\mathsf{BQP}^\mathsf{PH}$. Informally, in the $\OR \circ \textsc{Forrelation}$ problem, we are given an exponentially-long list $\{(f_k, g_k)\}_{k \in \{0,1\}^\kappa}$ of pairs of $n$-bit Boolean functions, and we must distinguish between:
\begin{itemize}
    \item[(YES)] There exists a single $k \in \{0,1\}^\kappa$ such that the functions $f_k$ and $g_k$ are \textit{Forrelated}, meaning that $\left|\braket{+^n | \Phi_{(f_k, g_k)}}\right| \ge \eps$ for some $\eps \ge 1/\poly(\kappa)$,\footnote{The Forrelation between $f$ and $g$ is usually defined directly in terms of the correlation between $\hat{f}$ and $g$ \cite{Aar10-bqp-ph, AA18-forrelation}, but our definition is equivalent. Indeed, $\left|\braket{+^n | \Phi_{(f_k, g_k)}}\right|^2$ is exactly the acceptance probability of the $2$-query quantum algorithm for estimating the Forrelation between $f$ and $g$ \cite[Section 3.2]{Aar10-bqp-ph}.} or
    \item[(NO)] For every $k \in \{0,1\}^\kappa$, $f_k$ and $g_k$ are uniformly random, in which case $\left|\braket{+^n | \Phi_{(f_k, g_k)}}\right|$ is negligible for every $k$, with high probability.
\end{itemize}

Our main insight is that viewing the $\textsc{Forrelation}$ problem as a state overlap problem allows us to relate $\OR \circ \textsc{Forrelation}$ to the $2$-Forrelation state PRS distinguishing task.
In fact, we will formally relate these problems via a reduction: we show that any $\mathsf{BQP}^\mathsf{PH}$ adversary that distinguishes the PRS ensemble from random would give rise to a $\mathsf{BQP}^\mathsf{PH}$ algorithm for solving $\OR \circ \textsc{Forrelation}$. Given an instance $\{(f'_k, g'_k)\}_{k \in \{0,1\}^\kappa}$ of $\OR \circ \textsc{Forrelation}$, we choose a uniformly random function $h: \{\pm 1\}^n \to \{\pm 1\}$. Then, the PRS adversary is given the input state $\ket{\Phi_h} = U_h \ket{+^n}$, and it is allowed queries to the oracle $\{(f_k, g_k)\}_{k \in \{0,1\}^\kappa}$ defined by $(f_k, g_k) \coloneqq (f'_k, g'_k \cdot h)$.

Observe that if $\{(f'_k, g'_k)\}_{k \in \{0,1\}^\kappa}$ is a YES instance of $\OR \circ \textsc{Forrelation}$, then there exists some $k$ such that
\[
\left|\braket{\Phi_h | \varphi_k}\right| = \left|\braket{\Phi_h | \Phi_{(f_k, g_k)}}\right| = \left|\braket{+^n | U_h U_h | \Phi_{(f'_k, g'_k)}}\right| = \left|\braket{+^n | \Phi_{(f'_k, g'_k)}}\right| \ge \eps.
\]
In other words, the state $\ket{\Phi_h}$ given to the adversary has some non-negligible overlap with a state $\ket{\varphi_k}$ drawn from the PRS ensemble. On the other hand, if $\{(f'_k, g'_k)\}_{k \in \{0,1\}^\kappa}$ is a NO instance of $\OR \circ \textsc{Forrelation}$, then $\ket{\Phi_h}$ is far from \textit{all} states in the PRS ensemble, with high probability.

Though this reduction does not perfectly map the $\OR \circ \textsc{Forrelation}$ problem onto the security challenge of the PRS, we nevertheless show that our reduction is quantitatively ``close enough,'' at least in the single-copy case. Specifically, we prove that the distinguishing advantage of the adversary between the YES and NO instances of $\OR \circ \textsc{Forrelation}$ is polynomially related to its distinguishing advantage between the pseudorandom and Haar-random security challenges of the PRS. This polynomial dependence scales with the parameter $\eps$.

Proving this dependence requires a delicate analysis based on a carefully constructed distributional version of the $\OR \circ \textsc{Forrelation}$ problem. Along the way, we introduce a new variant of the \textit{Forrelation distribution} \cite{Aar10-bqp-ph, RT}---i.e., a distribution over pairs of Boolean functions $(f, g)$ that are Forrelated with high probability. This distribution is defined as follows: first, we choose $f$ to be uniformly random. Then, independently for each $x \in \{\pm 1\}^n$, we sample $g(x) \in \{\pm 1\}$ with bias proportional to the Fourier coefficient $\hat{f}(x)$, modulo a rounding step in case $\left|\hat{f}(x)\right|$ is too large. Our Forrelation distribution has the advantage that it is discretely defined, in contrast to the distributions given by Aaronson \cite{Aar10-bqp-ph} or Raz and Tal \cite{RT}, which involve multivariate Gaussians. Thus, in some applications, it may be easier to analyze. We prove that our Forrelation distribution is pseudorandom against $\mathsf{AC^0}$ using techniques similar to \cite{CHHL18} based on polarizing random walks.

\subsection{Cryptographic Implications}
From a practical standpoint, a non-oracular version of our PRS construction can be instantiated by choosing the functions $\{(f_k, g_k)\}_{k \in \{0,1\}^\kappa}$ according to a cryptographic hash function, or from a pseudorandom function (PRF) ensemble keyed by $k$. This generalizes the construction of PRSs using phase states with PRF-chosen phases \cite{JLS18-prs, BS19-binary}. \Cref{thm:main_informal} then suggests that our construction based on Forrelation states is secure against a much broader class of attacks than phase state constructions: whereas PRF-based phase states can be distinguished from Haar-random by a $\mathsf{BQP}^\mathsf{NP}$ adversary \cite{Kre21-pseudorandom}, our construction remains secure even against the stronger class of $\mathsf{BQP}^\mathsf{PH}$ adversaries.\footnote{Strictly speaking, the $\mathsf{BQP}^\mathsf{NP}$ attack on phase states requires a polynomial number of copies of the state, while our proof of security against $\mathsf{BQP}^\mathsf{PH}$ adversaries only applies in the singe-copy case. However, we conjecture that the $t$-Forrelation construction remains secure even in the multi-copy case, at least for some sufficiently large $t$. We outline a plausible path towards proving this in \Cref{sec:open}.}

Alternatively, the proof of \Cref{thm:main_informal} can be understood as showing that there exists a cryptographically useful property of hash functions that is plausibly independent of the $\mathsf{P}$ vs.\ $\mathsf{NP}$ problem\comment{W: remove, answering a question of Ananth, Qian, and Yuen \cite[Section 1.2]{AQY22-prs}; L: Feel free to keep/remove if you want!}. Informally, this property is the following hardness assumption of a hash function $F$:

\begin{property}[\Cref{property}, informal]
    \label{property_informal}
    Let $F = \{(f_k, g_k)\}_{k \in \{0,1\}^\kappa}$ be a list of pairs of efficiently computable functions. Given quantum query access to an auxiliary function $h$, it is hard for an adversary to distinguish whether:
    \begin{enumerate}[(i)]
        \item There exists a $k$ such that $f_k$ is Forrelated with $g_k \cdot h$, or
        \item $h$ is uniformly random.
    \end{enumerate}
\end{property}

In other words, \Cref{property_informal} posits the hardness of detecting Forrelations between two parts $(f_k, g_k)$ of $F$ relative to a ``shift'' specified by $h$.
\comment{L: added}
Note that although this is a $\mathsf{QCMA}$-style problem in the sense that the shifted Forrelation is efficiently verifiable given the classical secret $k$, it is actually unclear whether it could be broken if $\mathsf{BQP} = \mathsf{QCMA}$.
This is because in this problem, an oracle of $h$ is given instead of (some succinct representation of) its code.

The key step in our proof involves reducing the security of the pseudorandom state ensemble to this problem, while also showing that a version of \Cref{property_informal} holds for the oracle $\mathcal{O}$ that we construct. As a result, we conclude that this property is simultaneously:

\begin{enumerate}[(a)]
    \item \label{item:intro_prs_sufficient} Powerful enough to construct various useful quantum cryptographic schemes, including commitments, zero knowledge, one-time signatures, etc., because it suffices to construct pseudorandom states,
    \item \label{item:intro_random_oracle} Plausibly true for existing hash functions like SHA-3, because it holds for a random oracle, and
    \item \label{item:intro_p_np_independent} Independent of the existence of one-way functions in the black-box setting (and, indeed, even independent of $\mathsf{P}$ vs.\ $\mathsf{NP}$).
\end{enumerate}

Prior to this work, we could only find properties that achieve any two of these three: the existence of one-way functions satisfies \eqref{item:intro_prs_sufficient} \cite{JLS18-prs,MY22-prs,AQY22-prs} and \eqref{item:intro_random_oracle} \cite{IR89-limits}; the quantum oracle constructed in \cite{Kre21-pseudorandom} achieves \eqref{item:intro_prs_sufficient} and \eqref{item:intro_p_np_independent}; and the trivial property satisfies \eqref{item:intro_random_oracle} and \eqref{item:intro_p_np_independent} \cite{BGS75-p-np}. Most notably, unlike the result of \cite{Kre21-pseudorandom}, we do not require a practical realization of Haar-random oracles in order to build OWF-independent pseudorandom states. For further discussion, see \Cref{sec:standard_model}.

\comment{OLD: Alternatively, the proof of \Cref{thm:main_informal} could be taken to suggest that there exists a cryptographically useful property of hash functions that is plausibly independent of the $\mathsf{P}$ vs.\ $\mathsf{NP}$ problem. Roughly speaking, the property says that given a list of functions $\{(f_k, g_k)\}_{k \in \{0,1\}^\kappa}$, it should be hard to find Forrelations between two parts of the function $(f_k, g_k)$ relative to a ``shift'' specified by a function $h$.
The key step in our proof involves reducing the security of the pseudorandom state ensemble to this problem, in the oracle setting. We believe that this could be turned into a concrete, non-oracular hardness assumption that is plausibly weaker than OWFs, and that suffices to construct PRSs, although we leave this to future work.}

\comment{
Our proof also requires other statistical properties of random oracles, but we believe this is a plausible approach to obtaining a concrete hardness assumption that is
  \begin{enumerate}[(i)]
    \item Plausibly true for existing hash functions like SHA-3, because it holds for a random oracle,
    \item Independent of $\mathsf{P}$ vs.\ $\mathsf{NP}$ in the black-box setting, and
    \item Powerful enough to construct various useful quantum cryptographic schemes, including commitments, zero knowledge, one-time signatures, etc.
  \end{enumerate}
In particular, unlike the result of \cite{Kre21-pseudorandom}, we satisfy (1)--(3) without the need for a practical realization of Haar-random oracles.
}

\comment{M: I am not sure I understand this discussion. Need to think about this more.} 

\comment{W: I like the updated version. It might be better if we also add to the end ``In particular, unlike the result of \cite{Kre21-pseudorandom}, we satisfy (1)--(3) without the need for a practical realization of Haar-random oracles.'' Also, I don't think we actually need to preface with ``if this reduction carries over''.}
\comment{W: maybe also add that we previously knew how to get (i) and (iii). Or more generally that previously, we could only get any two of these three.}

\subsection{Open Problems}
\label{sec:open}

Perhaps the most natural question left for future work is whether our result can be strengthened to an oracle relative to which $\mathsf{P} = \mathsf{NP}$ and \textit{multi-copy} PRSs exist. It seems reasonable to conjecture that $t$-Forrelation states should remain secure against $\mathsf{BQP}^\mathsf{PH}$ adversaries even in the multi-copy setting. However, our strategy based on reduction from the $\OR \circ \textsc{Forrelation}$ problem might not suffice to prove this, at least in the $t = 2$ case. For example, if we try to naively extend our reduction to the the multi-copy case, then the adversary receives $\ket{\Phi_h}^{\otimes T}$ for some arbitrary polynomially-bounded $T$, rather than a single copy of $\ket{\Phi_h}$. And though $\left|\braket{\Phi_h | \varphi_k}\right| \ge \eps$ may be non-negligible, $\left|\bra{\Phi_h}^{\otimes T} \ket{\varphi_k}^{\otimes T}\right| = \left|\braket{\Phi_h | \varphi_k}\right|^T \ge \eps^T$ could in general be negligible if $T$ is large enough.
As a result, the distinguishing advantage of the adversary between the YES and NO instances of $\OR \circ \textsc{Forrelation}$ may no longer be polynomially related to its distinguishing advantage between the pseudorandom and Haar-random security challenges of the PRS.

Nevertheless, we show that there is some hope in extending our approach to the multi-copy setting. Assuming a strong conjecture about $t$-Forrelation states for some $t = \poly(n)$, we conditionally prove, via techniques similar to the single-copy case, the $\mathsf{BQP}^\mathsf{PH}$-security of $t$-Forrelation states. We give a formal statement of the conjecture in \Cref{sec:conjectured_multi_security}. Roughly speaking, the conjecture posits that for any given $t$-Forrelation state $\ket{\Phi_G}$, it is hard for $\mathsf{AC^0}$ circuits of $2^{\poly(n)}$ size to distinguish a $t$-tuple of functions $F = (f^1, f^2, \ldots, f^t)$ chosen uniformly at random, from an $F$ chosen subject to the constraint that $\left|\braket{\Phi_F|\Phi_G}\right|$ is negligibly close to $1$. We expect that choosing $t$ to be a large polynomial would be necessary for this conjecture to hold, as otherwise there might be very few $F$s for which $\left|\braket{\Phi_F|\Phi_G}\right|$ is close to $1$.

Independent of the issue of single-copy versus multi-copy security, can our oracle be strengthened in other ways? For example, can one build an oracle relative to which $\mathsf{P} = \mathsf{NP}$ and \textit{pseudorandom unitaries} (PRUs) \cite{JLS18-prs} exist?\footnote{Actually, to our knowledge, it is open even to construct PRUs relative to \textit{any} classical oracle.}
Alternatively, could one give an oracle relative to which $\mathsf{P} = \mathsf{QMA}$ and PRSs exist? One challenge is that if multi-copy PRSs exist, then $\mathsf{P} \neq \mathsf{PP}$, as observed by Kretschmer \cite{Kre21-pseudorandom}. Thus, any oracle relative to which $\mathsf{P} = \mathsf{QMA}$ and multi-copy PRSs exist must also be an oracle relative to which $\mathsf{P} = \mathsf{QMA} \neq \mathsf{PP}$, which is still an open problem \cite{AIK21-acrobatics}. However, it is not clear whether a similar barrier exists in the single-copy case. Indeed, another important direction for future work is to better understand what computational assumptions are required to construct single-copy PRSs, and also seemingly weaker quantum cryptographic primitives such as EFI pairs \cite{BCQ23-efi}. In particular, does the existence of single-copy PRSs imply $\mathsf{P} \neq \mathsf{PSPACE}$?

Finally, we seek to better understand whether there is any sense in which Forrelation, or an oracle problem like it, is necessary for our construction.  Could the binary phase construction of pseudorandom states \cite{JLS18-prs,BS19-binary} in fact also be secure against $\mathsf{BQP^{PH}}$ adversaries in the single-copy setting? (Recall that in the multi-copy setting, it is insecure against $\mathsf{BQP^{PH}}$, and indeed $\mathsf{BQP^{NP}}$, as shown by Kretschmer \cite{Kre21-pseudorandom}.)

\section{Preliminaries}

\subsection{Basic Notation}

We denote by $[n]$ the set $\{1, 2, \ldots, n\}$. If $\mathcal{D}$ is a probability distribution, then $x \sim \mathcal{D}$ means that $x$ is a random variable sampled from $\mathcal{D}$. If $S$ is a finite set, then $x \sim S$ means that $x$ is a uniformly random element of $S$. We let $\ind \{P\}$ be the indicator function that evaluates to $1$ if the predicate $P$ is true, and $0$ otherwise.

We use $\log(x)$ to denote the base-$2$ logarithm of $x$, while $\ln(x)$ is the base-$e$ logarithm.
For two $n$-dimensional vectors $v, w$, $v \odot w$ denotes their Hadamard (entrywise) product, that is $(v \odot w)_i := v_i w_i$.
For $a>0$, we let $\trnc_{a}: \Reals \to [-a, a]$ be the function that truncates to the interval $[-a, a]$, i.e.\ $\trnc_a(z) \coloneqq \min \{a, \max \{-a, z\} \}$.
We may also omit the subscript $a$ if $a = 1$.
Observe that $\trnc_a(z) = a\cdot \trnc(z/a)$.

We use $\TVD(X, Y)$ to denote the total variation distance between probability distributions, and $\TD(\rho, \sigma)$ to denote the trace distance between quantum states. We denote by $\diag(X)$ the diagonal of a matrix $X$.

As in standard cryptographic notation, we use $\poly(n)$ to denote an arbitrary polynomially-bounded function of $n$, i.e.\ a function $f$ for which there is a constant $c > 0$ such that $f(n) \le n^c$ for all sufficiently large $n$. Likewise, we use $\polylog(n)$ for an arbitrary $f$ satisfying $f(n) \le \log(n)^c$ for all sufficiently large $n$, and $\quasipoly(n)$ for an arbitrary $f$ satisfying $f(n) \le 2^{\log(n)^c}$ for all sufficiently large $n$. $\negl(n)$ denotes an arbitrary negligibly-bounded function of $n$, i.e.\ a function $f$ with the property that for every $c > 0$, for all sufficiently large $n$, $f(n) \le n^{-c}$.

\subsection{Boolean Functions}
%

For convenience, we use the $\pm 1$ basis for Boolean functions. Every function $f: \{\pm 1\}^n \to \Reals$ can be represented uniquely as a real multilinear polynomial:
\[
f(x) = \sum_{S \subseteq [n]} \hat{f}(S) \cdot \prod_{i \in S} x_i,
\]
where $\hat{f}(S)$ are the Fourier coefficients of $f$. This allows us to extend the domain of $f$ to arbitrary inputs in $\Reals^n$.
The Fourier coefficients can be computed via:
\[
\hat{f}(S) = \frac{1}{2^n} \sum_{x \in \{\pm 1\}^n} f(x) \cdot \prod_{i \in S} x_i.
\]
for each $S \subseteq [n]$. In a slight abuse of notation, whenever $x \in \{\pm 1\}^n$, we let $\hat{f}(x) = \hat{f}(S)$ where $S \coloneqq \{i \in [n]: x_i = -1\}$. For $\ell \in [n]$, we denote by
\[
L_{1, \ell}(f) \coloneqq \sum_{S \subseteq [n] : |S| = \ell} \left|\hat{f}(S)\right|.
\]

If $f: \{\pm 1\}^n \to \{\pm 1\}$ is a Boolean function, we let $\ttable(f) \in \{\pm 1\}^{2^n}$ denote the truth table of $f$, i.e.\ the concatenation of $f$ evaluated on all possible Boolean inputs, in lexicographic order. We denote by $\mathsf{AC^0}[s, d]$ the set of Boolean circuits of size at most $s$ and depth at most $d$ consisting of unbounded fan-in $\AND$, $\OR$, and $\NOT$ gates.

\subsection{Concentration Inequalities}

The concentration inequalities stated below are standard (see e.g.\ {\cite[Chapter 2]{Ver18-hdp}}).

\begin{fact}[Hoeffding's inequality]
\label{fact:hoeffding}
Suppose $X_1,\ldots,X_n$ are independent random variables such that $X_i \in [a_i, b_i]$ for all $i$. Let $X = \sum_{i=1}^n X_i$ and let $\mu = \E[X]$. Then for all $s \ge 0$ it holds that:
\[
\Pr\mbracket{|X - \mu| \ge s} \le 2\exp\mparen{-\frac{2s^2}{\sum_{i=1}^n (b_i - a_i)^2}}.
\]
\end{fact}

A real-valued random variable $X$ is $\sigma$-subgaussian if $\Pr[|X- \mathbb{E}[X]| \ge s\sigma] \le 2e^{-s^2/2}$ holds for all $s \ge 0$. It follows from Hoeffding's inequality that for any vector $a \in \mathbb{R}^m$, the random variable $X = \langle{a},{Y}\rangle$ where $Y$ is uniform in $\{\pm 1\}^m$ is $\sigma$-subgaussian with $\sigma = \|a\|_2$. 

\begin{fact}[Bernstein's inequality]
\label{cor:subexp_concentration}
Let $X_1, \ldots, X_L$ be independent $\sigma$-subgaussian random variables. Let $X = \frac1L \sum_{i=1}^L X_i^2$ and let $\mu = \mathbb{E}[X]$. Then there exists an absolute constant $c > 0$ such that for every $s \ge 0$:
\[
\Pr\left[ \left| X - \mu \right| \ge s \right] \le 2\exp\left(
-c L \min \left\{\frac{s^2}{\sigma^4}, \frac{s}{\sigma^2} \right\}
\right).
\]
\end{fact}

We use the previous inequality to prove the following lemma.

\begin{lemma}\label{lem:fourier-concentration}
Let $f_1, \cdots, f_L: \{\pm 1\}^n \to \{\pm1\}$ be independent samples of uniformly random Boolean functions. Then, there exists an absolute constant $C > 0$, such that for any $z \in \{\pm 1\}^n$, we have 
    \[ \E\left|\frac1L \sum_{k=1}^L \hat{f_k}(z)^2 - \frac{1}{2^n}\right| \le \frac{C}{2^n \sqrt{L}}. \]
\end{lemma}
\begin{proof}
For a uniformly random Boolean function  $f: \{\pm 1\}^n \to \{\pm1\}$, the Fourier coefficient is given by $\hat{f}(z) = \langle{a},{Y}\rangle$ where $Y$ is uniform in $\{\pm 1\}^{2^n}$ and $\|a\|_2 = 2^{-n/2}$. Thus, $\hat{f}(z)$ is $\sigma$-subgaussian and $\mathbf{E}[|\hat{f}(z)|^2] = \sigma^2$ for $\sigma = 2^{-n/2}$. \Cref{cor:subexp_concentration} then implies the following tail bound
\[ \Pr\left[\left|\frac1L \sum_{k=1}^L \hat{f_k}(z)^2 - \frac{1}{2^n}\right| \ge s \right] \le 2\exp\left(
-c L \min \left\{\frac{s^2}{\sigma^4}, \frac{s}{\sigma^2} \right\}
\right). \]
Since $\mathbf{E}[X] = \int_0^\infty \Pr[X \ge s]ds$ for any non-negative random variable $X$, we have
\begin{align*}
    \ \E\left|\frac1L \sum_{k=1}^L \hat{f_k}(z)^2 - \frac{1}{2^n}\right|& \le 2\int_0^{\sigma^2} e^{-cLs^2/\sigma^4} ds + 2\int_{\sigma^2}^\infty e^{-cLs/\sigma^2}ds \\
    \ &=  \frac{2\sigma^2}{\sqrt{cL}} \int_0^{\sqrt{cL}} e^{-s^2}ds +  \frac{2\sigma^2}{{cL}} \int_{cL}^{\infty} e^{-s}ds \\
    &\le~ \frac{C\sigma^2}{\sqrt{L}},
\end{align*}
for an absolute constant $C>0$. Plugging in the value of $\sigma$ gives us the required bound.
\end{proof}

\subsection{Quantum States}

We also use the $\pm 1$ basis for quantum states. So, the space of $n$-qubit pure states is spanned by the orthonormal basis $\left\{\ket{x} : x \in \{\pm 1\}^n\right\}$, which we call the \textit{computational basis}. We denote by $\mu_{\mathrm{Haar}}^n$ the Haar measure over $n$-qubit pure states. We let $\ket{+} = \frac{\ket{1} + \ket{-1}}{\sqrt{2}}$, and use $\ket{+^n}$ as shorthand for $\ket{+}^{\otimes n}$.
When the context is clear, $H$ will generally denote the $n$-qubit Hadamard transform defined by
\[
H \coloneqq \begin{bmatrix}
\frac{1}{\sqrt{2}} & \frac{1}{\sqrt{2}}\\
\frac{1}{\sqrt{2}} & -\frac{1}{\sqrt{2}}
\end{bmatrix}^{\otimes n}.
\]
If $f$ is a Boolean function, we let $U_f$ be the phase oracle corresponding to $f$, i.e.\ the unitary transformation that acts as $U_f\ket{x} = f(x) \ket{x}$ on computational basis states $\ket{x}$. Note that, for any Boolean function $f$ and $x \in \{\pm 1\}^n$, $\bra{x}H U_f \ket{+^n} = \hat{f}(x)$.

We define multi-copy pseudorandom quantum states as follows.

\begin{definition}[Multi-copy pseudorandom quantum states \cite{JLS18-prs}]
\label{def:prs}
Let $\kappa \in \Naturals$ be the security parameter, and let $n(\kappa)$ be the number of qubits in the quantum system. A keyed family of $n$-qubit quantum states $\{\ket{\varphi_k}\}_{k \in \{0,1\}^\kappa}$ is \emph{multi-copy pseudorandom} if the following two conditions hold:
\begin{enumerate}[(i)]
\item (Efficient generation) There is a polynomial-time quantum algorithm $G$ that generates $\ket{\varphi_k}$ on input $k$, meaning $G(k) = \ket{\varphi_k}$.
\item (Computationally indistinguishable) For any polynomial-time quantum adversary $\mathcal{A}$ and every $T = \poly(\kappa)$:
\[
\left| \Pr_{k \sim \{0,1\}^\kappa}\mbracket{\mathcal{A}\mparen{1^\kappa, \ket{\varphi_k}^{\otimes T}} = 1 } - \Pr_{\ket{\psi} \sim \mu_{\mathrm{Haar}}^n}\mbracket{\mathcal{A}\mparen{1^\kappa, \ket{\psi}^{\otimes T}} = 1 } \right| \le \negl(\kappa).
\]
\end{enumerate}
\end{definition}

We emphasize that the above security definition must hold for \textit{all} polynomial values of $T$ (i.e.\ $T$ is not bounded in advance).\\

We also define single-copy pseudorandom states.
Unlike in the multi-copy case, we require $n > \kappa$ in order for the definition to be nontrivial, analogous to how a classical pseudorandom generator stretches a seed of length $\kappa$ into a pseudorandom string of length $n > \kappa$ (see \cite[Section 2.2]{MY22-prs} for further discussion).

\begin{definition}[Single-copy pseudorandom states \cite{MY22-prs}]
\label{def:single_copy_prs}
Let $\kappa \in \Naturals$ be the security parameter, and let $n(\kappa) > \kappa$ be the number of qubits in the quantum system. A keyed family of $n$-qubit quantum states $\{\ket{\varphi_k}\}_{k \in \{0,1\}^\kappa}$ is \emph{single-copy pseudorandom} if the following two conditions hold:
\begin{enumerate}[(i)]
\item (Efficient generation) There is a polynomial-time quantum algorithm $G$ that generates $\ket{\varphi_k}$ on input $k$, meaning $G(k) = \ket{\varphi_k}$.
\item (Computationally indistinguishable) For any polynomial-time quantum adversary $\mathcal{A}$:
\[
\left| \Pr_{k \sim \{0,1\}^\kappa}\mbracket{\mathcal{A}\mparen{1^\kappa, \ket{\varphi_k}} = 1 } - \Pr_{\ket{\psi} \sim \mu_{\mathrm{Haar}}^n}\mbracket{\mathcal{A}\mparen{1^\kappa, \ket{\psi}} = 1 } \right| \le \negl(\kappa).
\]
\end{enumerate}
\end{definition}

In this work, we only consider uniform quantum adversaries. That is, the adversary $\mathcal{A}$ is specified by a polynomial-time Turing machine $M$, where $M(1^\kappa)$ outputs a quantum circuit that implements $\mathcal{A}$ on security challenges of size $\kappa$. However, our construction is plausibly secure against non-uniform adversaries as-is, or possibly via the addition of a salting step \cite{CGLQ20-tradeoffs}.

\section{Pseudorandomness of the Forrelation Distribution}

We define our discrete version of the Forrelation distribution as follows:

\begin{definition}
\label{def:forrelation}
The \emph{Forrelation distribution} $\mathcal{F}_n$ is a distribution over a pair of functions $f, g : \{\pm 1\}^n \to \{\pm 1\}$ sampled as follows. First, $f$ is sampled uniformly at random. Then, for each $x \in \{\pm 1\}^n$, $g(x)$ is sampled independently via:
\[
g(x) = \begin{cases}
1 & \text{\rm with probability } \frac{1 + \trnc\left(\sqrt{\eps 2^n} \hat{f}(x)\right)}{2}\\
-1 & \text{\rm with probability } \frac{1 - \trnc\left(\sqrt{\eps 2^n} \hat{f}(x)\right)}{2},
\end{cases}
\]
where $\eps = \frac{1}{100n}$.
\end{definition}


The main technical result of this section is that the above Forrelation distribution $\mathcal{F}_n$ is pseudorandom against all constant-depth $2^{\poly(n)}$-size $\mathsf{AC^0}$ circuits.
\begin{theorem}
\label{thm:forrelation_AC0_pseudorandom}
For every $C \in \mathsf{AC^0}[2^{\poly(n)}, O(1)]$, the Forrelation distribution $\mathcal{F}_n$ satisfies
\[
\left|\E_{(f, g) \sim \mathcal{F}_n}\left[C(\ttable(f), \ttable(g))\right] - \E_{z \sim \{\pm 1\}^{2 \cdot 2^{n}}}\left[C(z)\right]\right| \le \frac{\poly(n)}{\sqrt{2^n}}.
\]
\end{theorem}

\Cref{thm:forrelation_AC0_pseudorandom} will be a special case of the following theorem.
The special case holds since for any $C\in \mathsf{AC^0}[2^{\poly(n)}, O(1)]$ it holds that $L_{1,2}(C) \le \poly(n)$, as proved in \cite{Tal17}.
\begin{theorem}\label{thm:new_BQP_PH_General}
Let $N= 2^n$. Let $\mathcal{C}$ be a family of $2N$-variate Boolean functions, which is closed under restrictions. 
Assume that for any $C\in \mathcal{C}$ it holds that $L_{1,2}(C)\le t$.
Then, for any $C\in \mathcal{C}$ it holds that 
$$\left|\E_{f,g\sim \mathcal{F}_n}[C(\ttable(f),\ttable(g))] - \E[C]\right| \le O\!\left(\frac{t\cdot \log N}{\sqrt{N}}\right).$$
\end{theorem}

\begin{proof}
We show how to obtain the distribution $\mathcal{F}_n$ approximately as a result of a random walk, taking $\polylog(N)$ steps, where each step is a multi-variate Gaussian.
\begin{enumerate}
\item Let $m = 200 \ln(N)/\eps$.
\item Let $X^{(\le 0)} = \vec{0}$, $Y^{(\le 0)} = \vec{0}$.
\item For $i = 1, \ldots, m$:
 \begin{align*}\text{Let~}&X^{(i)}\sim N(0,\eps)^{N}\\
 \text{Let~}& D =(1-|X^{(\le i-1)}|)\\
 &X^{(\le i)} = X^{(\le i-1)} + D \odot \trnc(X^{(i)})\\
  &Y^{(\le i)} =  \trnc(Y^{(\le i-1)} + \trnc_{1/2}(\sqrt{\eps} H D \odot X^{(i)}))\end{align*}
  \item Output $(X^{(\le m)}, Y^{(\le m)})$.
\end{enumerate}
We observe that by definition, the coordinates of $X^{(\le i)}$ and $Y^{(\le i)}$ are bounded in $[-1,1]$, and the coordinates of $X^{(\le i)}$ are independent. We further make the following three claims/observations.
	
The first claim should be interpreted as ``with high probability, truncations are irrelevant''.
\begin{claim}\label{claim:Y} 
With probability at least $1-m/N^{10}$, for all $i\in [m]$\\ 
$$Y^{(\le i)} = \sqrt{\eps} H X^{(\le i)} \text{~~~and~~~}  Y^{(\le i)} \in [-1/2,1/2]^{N}.$$
\end{claim}

The second claim should be interpreted as ``$X^{(\le i)}$ polarizes'', i.e., coordinates get closer to $\pm1$.
\begin{claim}\label{claim:X_polarize}
With probability at least $1-1/N^2$, we have 
$$|X^{(\le m)}| \in [1-1/N^3, 1]^N.$$
\end{claim}
	
The third claim should be interpreted as ``with high probability, the change under the $i$-th step is small (with respect to $C$)''.
\begin{claim}\label{claim:small_step}
For any $i\in [m]$, with probability at least $1-2/N^2$ over $(X^{(\le i-1)}, Y^{(\le i-1)})$ 
(i.e., the history before step $i$), 
the $i$-th step size satisfies	
$$\E_{X^{(i)}}\mbracket{C\mparen{X^{(\le i)},Y^{(\le i)}}-C\mparen{X^{(\le i-1)}, Y^{(\le i-1)}}} \le O\mparen{t\eps /\sqrt{N}}.$$
\end{claim}
We defer the proof of the claims to \Cref{sec:proof_of_claims}. We show how to complete the proof given the three claims. 
Let $\mathcal{E}$ be the event that:
\begin{align}
&Y^{(\le m)} = \sqrt{\eps} H X^{(\le m)} \label{cond:Y}\\
&Y^{(\le m)} \in [-1/2,1/2]^N \label{cond:Y2}\\
&\abs{X^{(\le m)}} \in [1-1/N^3, 1]^N \label{cond:X}\\
\forall i\in [m]:\;&  \left(X^{(\le i-1)}, Y^{(\le i-1)}\right)\text{~satisfies~}\nonumber\\
&\E_{X^{(i)}}\mbracket{C\mparen{X^{(\le i)},Y^{(\le i)}}-C\mparen{X^{(\le i-1)}, Y^{(\le i-1)}}} \le O\mparen{t\eps/\sqrt{N}}.
\label{cond:small_step}
\end{align}
Let $\delta := \Pr[\neg \mathcal{E}]$ which by the three claims is at most $1/N$ for sufficiently large $N$.
For each $i\in [m]$, we have 
$$\left|\E[C(X^{(\le i)}, Y^{(\le i)}) \mid \mathcal{E}] - \E[C(X^{(\le i-1)}, Y^{(\le i-1)})\mid \mathcal{E}]\right| \le O\mparen{\delta + t\eps/\sqrt{N}}\le O\mparen{t\eps/\sqrt{N}},$$
since conditioned on $\mathcal{E}$, the $(i-1)$-th history definitely satisfies Condition~\eqref{cond:small_step}, but the conditioning might change the distribution of $X^{(i)}$ by up to $\delta$ total-variation distance, and we need to compensate for that.
We get that 
$$\left|\E[C(X^{(\le m)}, Y^{(\le m)}) \mid \mathcal{E}] - \E[C]\right| \le O\mparen{mt\eps/\sqrt{N}}.$$
From Condition~\eqref{cond:X}, we see that $X' := \sgn(X^{\le m})$ is $1/N^3$-close to $X^{(\le m)}$. By Condition~\eqref{cond:Y} we see that $Y^{(\le m)} = \sqrt{\eps} H X^{(\le m)}$ and thus 
$$Y' := \sqrt{\eps} H X' = \sqrt{\eps} H X^{(\le m)} + \sqrt{\eps} H (X' - X^{(\le m)}) = Y^{(\le m)} + err$$ where each coordinate of $err$ is at most $\sqrt{\eps} \sqrt{N}/N^3 \le 1/N^2$ in absolute value.
By Condition~\eqref{cond:Y2}, we get that $Y' \in [-1,1]^N$.
We apply the next claim to get 
\[\left|\E[C(X',Y') | \mathcal{E}] - \E[C(X^{(\le m)}, Y^{(\le m)}) | \mathcal{E}]\right| \le 2/N.\]
\begin{fact}[\protect{Folklore, See for example \cite[Lemma~2.7]{CHLT}}]\label{claim:close_means_close}
Let $C:[-1,1]^{2N} \to [-1,1]$ be a multi-linear function. Then for every $z, z'\in [-1,1]^{2N}$ we have $|C(z)-C(z')|\le 2N\cdot \|z-z'\|_{\infty}$.
\end{fact}
Then, triangle inequality gives
$$\abs{\E[C(X',Y') | \mathcal{E}] - \E[C]} \le 2/N + O\mparen{mt\eps/\sqrt{N}} \le O\mparen{mt\eps/\sqrt{N}},$$
and since $(X', Y')|\mathcal{E}$ is $\delta$-close in statistical distance to the distribution $(X', \trnc(\sqrt{\eps} H X'))$ we get
$$|\E[C(X',  \trnc(\sqrt{\eps} H X'))] - \E[C]| \le O\mparen{mt\eps/\sqrt{N}} + \delta \le O\mparen{mt\eps/\sqrt{N}} = O\mparen{t\log(N)/\sqrt{N}}.$$
Finally, we observe that $X'$ is the uniform distribution over $\{-1,1\}^N$ and the expectation of any multilinear polynomial under $(X', \trnc(\sqrt{\eps} H X'))$ is the same as that under the Forrelation distribution $\mathcal{F}_n$. \comment{W: is this sufficiently clear?}
\end{proof}

\subsection{Proofs of the Three Claims}\label{sec:proof_of_claims}
We will rely on the following theorem from~\cite{RT} and the following lemma from~\cite{CHHL18}.
\begin{theorem}[\protect{\cite{RT}, as restated in \cite[Theorem 9]{CHLT}}]\label{thm:RT} Let $n, t\ge 1$, $\delta \in (0,1)$.
Let $Z \in \Reals^n$ be a zero-mean multivariate Gaussian random variable with the following two properties:
\begin{enumerate}
\item For $i \in [n]$: $\var[Z_i] \le \frac1{8\ln(n/\delta)}$.
\item	For $i,j \in [n], i\neq j: |\cov[Z_i,Z_j]| \le \delta$.
\end{enumerate}
Let $\mathcal{C}$ be a family of $n$-variate Boolean functions, which is closed under restrictions. 
Assume that $L_{1,2}(\mathcal{C})\le t$. 
Then, for any $C\in \mathcal{C}$ it holds that 
$\left|\E\mbracket{C(\trnc(Z))} - C(\vec{0})\right| \le O(\delta \cdot t)$.
\end{theorem}

\begin{claim}[\protect{\cite[Claim 3.3]{CHHL18}]}] \label{claim:CHHL}
	Let $f$ be a multilinear function on $\Reals^n$ and 
	$v\in (-1, 1)^n$. Let $\delta \in [0,1]^n$ with $\delta_i \le 1-|v_i|$. Then, there exists a distribution over random restrictions $\rho$ such that for any $z\in \Reals^n$,
	$$
	f(v+\delta \odot z)-f(v) = \E_{\rho}[f_{\rho}(z) - f_{\rho}(\vec{0})].$$
\end{claim}

We remark that the statement of \cite[Claim 3.3]{CHHL18} as stated in their paper is slightly different from the above, but the above claim is implicit in their proof. 
\comment{L: It is very hard for me to see why this follows from CHHL Claim 3.3... The claim statements seem different?} \comment{~M: It's what the proof of Claim 3.3 gives you. $F$ in that proof is the distribution over $f_\rho$ defined here. L: It would be great if we could clarify this point in the text for readers unfamiliar with CHHL like myself :)}\comment{M: How about this? L: Is it easy to reprove this, maybe in the appendix, just for completeness? It seems like the proof of Claim 3.3 is short anyways. Also how is ``random restrictions'' defined? Isn't ``distribution'' already implying some randomness? Is it necessary to say a ``distribution over random restrictions''?}

We go on to prove the three claims. We start with the proof of \Cref{claim:small_step} as it is the hardest.
\begin{proof}[Proof of \Cref{claim:small_step}]
Fix a history $x^{(\le i-1)}, y^{(\le i-1)}$ and assume that $y^{(\le i-1)} \in [-1/2,1/2]^{N}$, an event which happens with probability at least $1-1/N^2$.
Let $d = (1-|x^{(\le i-1}|)$.
By definition, 
$$Y^{(\le i)} = \trnc\mparen{y^{(\le i-1)} + \trnc_{1/2}(\sqrt{\eps} H (d \odot X^{(i)})},$$ 
and by our assumption on $y^{(\le i-1)}$ 
we get that we can get rid of the outer $\trnc$.
Furthermore we observe that $\trnc_{1/2}(x) = \frac{1}{2}\trnc(2x)$ 
and we can thus simplify further to
$$Y^{(\le i)} = y^{(\le i-1)} + \frac{1}{2}\cdot \trnc\mparen{2\sqrt{\eps} H (d \odot X^{(i)})}.$$ 
We plug this in and apply \Cref{claim:CHHL}, to get that LHS of \Cref{claim:small_step}
\begin{align*}
&\equad \E_{X^{(i)}}\mbracket{C\mparen{X^{(\le i)},Y^{(\le i)}}-C\mparen{x^{(\le i-1)}, y^{(\le i-1)}}} \\
&= \E_{X^{(i)}}\mbracket{C\!\left(
x^{(\le i-1)} + d \odot \trnc\mparen{X^{(i)}}, \;
y^{(\le i-1)} + \frac{1}{2}\cdot \trnc\mparen{2\sqrt{\eps} H (d \odot X^{(i)})}\right) 
- C\!\left(x^{\le (i-1)}, y^{\le (i-1)}\right)} \\
&= \E_{X^{(i)},\rho}\mbracket{C_{\rho}\mparen{\trnc\mparen{X^{(i)}}, \trnc\mparen{2\sqrt{\eps} H (d \odot X^{(i)})}} - C_{\rho}(0)}
\end{align*}
for some distribution over random restrictions $\rho$.
To apply \Cref{thm:RT}, it remains is to bound the variances and co-variances of the coordinates in $\mparen{X^{(i)}, 2\sqrt{\eps} H (d \odot X^{(i)})}$.
We see that:
\begin{align*}
\forall{j}:\;
&\var\mparen{X^{(i)}_j} = \eps\\
\forall{j}:\;
&\var\mparen{(2\sqrt{\eps} H (d \odot X^{(i)}))_j} = 4\eps \cdot \eps \sum_{\ell}H_{j,\ell}^2 \cdot d_{\ell}^2  \le 4\eps^2\\
\forall{j\neq k}:\;
&\cov\mparen{X^{(i)}_j, X^{(i)}_k} = 0\\
\forall{j, k}:\;
&\abs{\cov\mparen{2\sqrt{\eps} H (d \odot X^{(i)}))_j, X^{(i)}_k}} = \abs{2\sqrt{\eps} \cdot H_{j,k} \cdot d_k \cdot \eps} \le  2\eps^{3/2}/\sqrt{N}\\
\forall{j \neq k}:\;&
\abs{\cov\mparen{(2\sqrt{\eps} H (d \odot X^{(i)}))_j, (2\sqrt{\eps} H (d \odot X^{(i)}))_k}} =4 \eps^2 \cdot \left|\sum_{\ell} H_{j, \ell} H_{k, \ell} d_{\ell}^2\right|
\end{align*}
To bound the last term, we consider the history $X^{(\le i-1)}$ as a random variable.
We denote by $D = (1-|X^{(\le i-1)}|)$ and show that with high probability over the history, $4 \eps^2 \cdot \left|\sum_{\ell} H_{j, \ell} H_{k, \ell} D_{\ell}^2\right|$ is small.
Observe that since the rows of $H$ are orthogonal, and all entries are $\pm 1/\sqrt{N}$, then there are exactly $N/2$ indices $\ell$ such that $H_{j, \ell} H_{k,\ell} = 1/N$ and exactly $N/2$ indices $\ell$ such that $H_{j, \ell} H_{k,\ell} = -1/N$. Thus, we can pair each positive index $\ell$ with a negative index $\ell'$. For each such pair, $(\ell, \ell')$, the expectation of  $\frac{1}{N}(D_{\ell}^2 - D_{\ell'}^2)$ is $0$ as the coordinates of $X^{(\le i-1)}$ are i.i.d. Furthermore, $\frac{1}{N}(D_{\ell}^2 - D_{\ell'}^2)$ is bounded in $[-1/N,1/N]$. 
By Hoeffding's inequality (\Cref{fact:hoeffding}), we get that for any $\eta$,
$$
\Pr\mbracket{\left|\sum_{\ell} H_{j, \ell} H_{k, \ell} D_{\ell}^2\right| \ge \eta} \le 2\exp(-\eta^2 N).
$$
By taking $\eta := 1/\sqrt{\eps N}$, we get that this probability is smaller than $2\exp(-1/\eps) = 2/N^{100}$.
To summarize, we see that with probability at least $1-2/N^2$, the history satisfies $Y^{(\le i-1)} \in [-1/2,1/2]^{N}$ and
$\left|\sum_{\ell} H_{j, \ell} H_{k, \ell} D_{\ell}^2\right| < 1/\sqrt{\eps N}$, making the co-variances of $(X^{(i)}, 2\sqrt{\eps} H (d \odot X^{(i)}))$ smaller than $\delta := \eps/\sqrt{N}$ and the variances smaller than $\eps \le 1/(8\ln(N/\delta))$, 
and hence \Cref{thm:RT} gives 
\[
\forall{\rho}: \E_{X^{(i)}}\mbracket{C_{\rho}\mparen{\trnc(X^{(i)}), \trnc(2\sqrt{\eps} H (d \odot X^{(i)}))} - C_{\rho}(0)} \le O(t \cdot \delta) = O(t \eps/\sqrt{N})\;.
\qedhere\]
\end{proof}

\begin{proof}[Proof of \Cref{claim:Y}]

For $i = 1, \ldots, m$, let $\mathcal{E}_i$ be the event that $Y^{(\le i)}$ satisfy the conditions of the claim, i.e., $Y^{(\le i)} = \sqrt{\eps} H X^{(\le i)}$ and 
	$Y^{(\le i)} \in [-1/2,1/2]^{N}$.
	
	We prove by induction on $i$ that $\Pr[\mathcal{E}_1 , \dots , \mathcal{E}_i] \ge 1-i/N^{10}$.
	The claim surely holds for $i=0$.
	For $i\ge 1$, 
		conditioned on $\mathcal{E}_{1}, \ldots, \mathcal{E}_{i-1}$, $X^{(i)}$ is still completely random, and by the concentration of multi-variate Gaussians, the $i$-th step size is small with high probability. 
		More precisely, $\sqrt{\eps} H (D \odot X^{(i)})$ gives $N$ independent Gaussians with zero-mean and variance $\le \eps^2$. They are all in the range $[-1/2, 1/2]$ with a probability of at least 
		$1-2N\cdot \exp(-\Omega(1/\eps^2))\ge 1-\negl(N)$.
		Furthermore $X^{(i)} \in [-1,1]^N$ with  with probability at least $1-2N \exp(-1/2\eps) \ge 1-1/N^{40}$.
		Let $\mathcal{E}'_{i}$ be the event that all coordinates of $\sqrt{\eps} H (D \odot X^{(i)})$ are between $-1/2$ and $1/2$ and all coordinates of $X^{(i)}$ are between $-1$ and $1$.
		Under this event, we get that the truncations do nothing, and we have 
		\begin{align*}
		X^{(\le i)} &= X^{(\le i-1)} + D \odot X^{(i)}\\
		Y^{(\le i)} &= Y^{(\le i-1)} + \sqrt{\eps} H (D \odot X^{(i)}).
		\end{align*}
Thus, under $\mathcal{E}_{1}, \ldots, \mathcal{E}_{i-1}, \mathcal{E}'_{i}$, we have
		$$Y^{(\le i)} = \sqrt{\eps} H X^{(\le i)}$$
		which concludes the proof of the first property.
		
		As for the second property, observe that without any conditioning, $X^{(\le i)}$ is a collection of $N$ i.i.d. zero-mean bounded random variables in $[-1,1]$.
		This means that we can apply Hoeffding's inequality (\Cref{fact:hoeffding}) to conclude that 
		$$\forall{j\in [N]}:\;\; \Pr[|(\sqrt{\eps} H X^{(\le i)})_j| \ge 1/2] \le 2\cdot \exp(-1/(8\eps)).$$
		Finally, we observe that if both $Y^{(\le i)} = \sqrt{\eps} H X^{(\le i)}$ and $(\sqrt{\eps} H X^{(\le i)}) \in [-1/2, 1/2]^N$ happen, then the event $\mathcal{E}_i$ happens.
		Taking a union bound over the bad events, we get that 
		$$
		\Pr[\neg \mathcal{E}_i \;|\; \mathcal{E}_1 , \dots  , \mathcal{E}_{i-1}] \le \negl(N) + 1/N^{40} + 2N \cdot \exp(-1/(8\eps)) \le 1/N^{10}\;,
		$$
		which, in turn, implies that 
		$\Pr[\mathcal{E}_1, \ldots, \mathcal{E}_i] \ge 1-i/N^{10}$.
					\end{proof}

To prove \Cref{claim:X_polarize} we rely on the following Claim from~\cite{CHHL18}.
\begin{claim}[\protect{\cite[Claim~3.5]{CHHL18}}]\label{claim:CHHL_polarizing}
Let $A_1, \ldots, A_m \in [-1,1]$ be independent symmetric random variables with $\E[A_{i}^2] \ge p$. For $i=1, \ldots, m$ define $B_i = B_{i-1} + (1-|B_{i-1}|) A_i$, where $B_{0} = 0$. Then, $\E[B_m^2] \ge 1-q$ for $q = 3\exp(-mp/16)$.
\end{claim}

\begin{proof}[Proof of \Cref{claim:X_polarize}]
Fix $j\in [N]$.
We see that the sequences $A_1 = \trnc(X^{(1)}_j), \ldots, A_m = \trnc(X^{(m)}_j)$ and $B_1 = X^{(\le 1)}_j, \ldots, B_m = X^{(\le m)}_j$ satisfy $B_i = B_{i-1} + (1-|B_{i-1}|) A_i$.
We also have that $A_1, \ldots, A_m$ are independent, symmetric random variables with $$\E[A_i^2] =  \E\mbracket{(X^{(i)}_j)^2} - \E\mbracket{\left((X^{(i)}_j)^2-1\right) \cdot \ind\left\{|X^{(i)}_j|\ge 1\right\}} \ge \eps - 1/N^{50}\ge \eps/2\;.$$

Thus, we get that $\E[1-|B_m|] \le \E[1-B_m^2] \le q$ for $q = 3\exp(-m\eps/32) \le 1/N^{6}$, and in particular $\Pr[1-|B_m| \ge  1/N^{3}] \le 1/N^{3}$.
Taking union bound over all $N$ coordinates completes the proof.
\end{proof}

\section{Construction of the Oracle}

The oracle used in our construction of pseudorandom states is simple to describe: it consists of a uniformly random oracle $A$, and an oracle $B$ that answers all $\mathsf{NP}$ queries to $A$ or $B$. Formally, we construct the oracle as follows.

\begin{definition}
\label{def:PH_oracle}
For a language $A: \{\pm 1\}^* \to \{\pm 1\}$, we define a language $\mathcal{O}[A]$ as follows. We construct an oracle $B$ inductively: for each $\ell \in\mathbb{N}$ and $x \in \{\pm 1\}^\ell$, view $x$ as an encoding of a pair $\langle M, y \rangle$ such that

\begin{enumerate}
\item $\langle M, y \rangle$ takes less than $\ell$ bits to specify,\footnote{Note that there are $2^\ell - 1$ such possible $\langle M, y \rangle$, which is why we take an encoding in $\{\pm 1\}^\ell$.}
\item $M$ is an $\mathsf{NP}$ oracle machine and $y$ is an input to $M$,
\item $M$ is syntactically restricted to run in less than $\ell$ steps, and to make queries to $A$ and $B$ on strings of length at most $\lfloor \sqrt{\ell} \rfloor$.
\end{enumerate}
Then we define $B(x) \coloneqq M(y)$. Finally, let $\mathcal{O}[A] = (A, B)$.
\end{definition}

Oracles such as those defined in \Cref{def:PH_oracle} always collapse $\mathsf{NP}$ to $\mathsf{P}$, as shown below.

\begin{proposition}
\label{prop:oracle_p_equals_np}
For any language $A: \{\pm 1\}^* \to \{\pm 1\}$, $\mathsf{P}^{\mathcal{O}[A]} = \mathsf{NP}^{\mathcal{O}[A]}$.
\end{proposition}

\begin{proof}
Given an $\mathsf{NP}^{\mathcal{O}[A]}$ machine $M$ and input $y$, a polynomial time algorithm can decide $M(y)$ by taking $x = \langle M, y \rangle$ and querying $B(x)$.
\end{proof}

Similar oracle constructions appeared in \cite{BM99-relativized, AIK21-acrobatics}. Morally speaking, queries to $\mathcal{O}[A]$ are roughly equivalent in power to queries to $\mathsf{PH}^A$. Indeed, any $\mathsf{PH}^A$ language can be decided in $\mathsf{P}^{\mathcal{O}[A]}$, by a simple extension of \Cref{prop:oracle_p_equals_np}. A partial converse also holds: via the well-known connection between between $\mathsf{PH}$ algorithms and $\mathsf{AC^0}$ circuits \cite{FSS84-circuit-oracle}, each bit of $\mathcal{O}[A]$ can be computed by an exponential-sized $\mathsf{AC^0}$ circuit depending on $A$ (see \Cref{lem:PH_to_AC0} in \Cref{app:proof_of_h2_h3} for a precise statement).


We next define the quantum states that we use to construct pseudorandom ensembles relative to our oracles, which are based on $t$-Forrelation states.

\begin{definition}[$t$-Forrelation states]
\label{def:t_forrelation_states}
For a $t$-tuple of functions $F = (f^1, f^2, ..., f^t)$ where $f^i : \{\pm 1\}^n \to \{\pm 1\}$, we denote by $\ket{\Phi_F}$ the state:
\[
\ket{\Phi_F} \coloneqq U_{f^t} \cdot H \cdot U_{f^{t-1}} \cdot H \cdots H \cdot U_{f^1} \ket{+^n}.
\]
where $U_{f^i}$ is the unitary phase oracle corresponding to $f^i$ and $H$ is the $n$-qubit Hadamard transform. We call any such state a \emph{$t$-Forrelation state}.
\end{definition}

The pseudorandom state ensembles we consider consist of random $t$-Forrelation states where the phase oracles are specified by the random oracle $A$.

\begin{definition}[State ensemble relative to $A$]
\label{def:prs_oracle_ensemble}
Fix a security parameter $\kappa$ and $t \ge 1$. We define an ensemble of $n$-qubit states for some $\kappa + 1 \le n \le \poly(\kappa)$. For each $k \in \{0,1\}^\kappa$ and $i \in [t]$, define $f_k^i: \{\pm 1\}^n \to \{\pm 1\}$ by $f_k^i(x) = A(x, k, i)$.\footnote{In a slight abuse of notation, $k$ and $i$ correspond to their binary representations over $\{\pm 1\}^\kappa$  and $\{\pm 1\}^{\lceil \log t \rceil}$, respectively, in $A(x, k, i)$.} Letting $F_k = (f^1_k, f^2_k, \ldots, f^t_k)$ we choose:
\[
\ket{\varphi_k} \coloneqq \ket{\Phi_{F_k}},
\]
and take the ensemble to be $\{\ket{\varphi_k}\}_{k \in \{0,1\}^\kappa}$.
\end{definition}

The main goal of the remainder of this work will be to show that when $A$ is a random oracle, then with probability $1$ over $A$, the set $\{\ket{\varphi_k}\}_{k \in \{0,1\}^\kappa}$ forms a secure pseudorandom state ensemble relative to $\mathcal{O}[A]$. We emphasize that our proofs will show security for \textit{any} function $n(\kappa)$ that satisfies $\kappa + 1 \le n \le \poly(\kappa)$.



\section{Single-Copy Security}

\label{sec:single_security}

Throughout this section, we fix $t = 2$ in \Cref{def:prs_oracle_ensemble}. Additionally, we will always denote $(f_k^1, f_k^2)$ by $(f_k, g_k)$, so that:
\[
\ket{\varphi_k} = U_{g_k} \cdot H \cdot U_{f_k} \ket{+^n}.
\]
The goal of this section is to prove, relative to $\mathcal{O}[A]$, the pseudorandomness of the ensemble $\{\ket{\varphi_k}\}_{k \in \{0,1\}^\kappa}$ defined in \Cref{def:prs_oracle_ensemble}.

\begin{theorem}
\label{thm:single_copy_prs_prob_1}
With probability $1$ over a random oracle $A$, $\{\ket{\varphi_k}\}_{k \in \{0,1\}^\kappa}$ is single-copy pseudorandom relative to $\mathcal{O}[A]$.
\end{theorem}



\subsection{Construction of Hybrids}

We will prove \Cref{thm:single_copy_prs_prob_1} via a hybrid argument. Each hybrid below defines a security challenge for the quantum adversary. The security challenge consists of a state $\ket{\psi}$ and an oracle $A$ that are sampled by the hybrid (note that $\ket{\psi}$ may in general depend on $A$). The adversary is given a single copy of $\ket{\psi}$ as input, and can make queries to $\mathcal{O}[A]$.

For convenience, in each of these hybrids we only specify the part of $A$ that corresponds to the functions $\{(f_k, g_k)\}_{k \in \{0,1\}^\kappa}$ that are used to construct the states with security parameter $\kappa$. Recall that $f_k(x) =  A(x, k, 1)$ and $g_k(x) =  A(x, k, -1)$. Otherwise, the rest of $A$ is always sampled uniformly at random. 

Ultimately, we wish to show the indistinguishability of the following two challenges:

\paragraph{Hybrid $\mathsf{H_0}$:} Sample $k^* \sim \{0,1\}^\kappa$. For each $k \in \{0,1\}^\kappa$, sample $f_k, g_k : \{\pm 1\}^n \to \{\pm 1\}$ uniformly at random. The adversary gets $\ket{\psi} = \ket{\varphi_{k^*}}$ as input.

\paragraph{Hybrid $\mathsf{H_4}$:} For each $k \in \{0,1\}^\kappa$, sample $f_k, g_k : \{\pm 1\}^n \to \{\pm 1\}$ uniformly at random. The adversary gets a Haar-random state $\ket{\psi}$ as input.\\

Hybrid $\mathsf{H_0}$ corresponds to sampling a state from the PRS ensemble, whereas Hybrid $\mathsf{H_4}$ corresponds to sampling a Haar-random state. We will interpolate between these hybrids by changing how we sample either the oracle $A$ or the state $\ket{\psi}$ in each step.

In the first intermediate hybrid, we observe that the uniform distribution over $(f_k, g_k)$ can also be generated by first sampling a Forrelated $(f'_k, g'_k)$, and then multiplying $g'_k$ pointwise with a uniformly random function. This motivates the next hybrid, which we shall show is equivalent to $\mathsf{H}_0$.

\paragraph{Hybrid $\mathsf{H_1}$:}
Sample $k^* \sim \{0,1\}^\kappa$. For each $k \in \{0,1\}^\kappa$, sample $f'_k, g'_k : \{\pm 1\}^n \to \{\pm 1\}$ as follows:
\begin{itemize}
\item If $k = k^*$, draw $(f'_k, g'_k) \sim \mathcal{F}_{n}$.
\item If $k \neq k^*$, draw $f'_k, g'_k$ uniformly at random.
\end{itemize}
Additionally, sample a random function $h: \{\pm 1\}^n \to \{\pm 1\}$. For each $k \in \{0,1\}^\kappa$, set $f_k = f'_k$ and $g_k = g'_k \cdot h$ (i.e.\ XOR in the $\pm 1$ domain). The adversary gets $\ket{\psi} = \ket{\varphi_{k^*}}$ as input.\\

From the results of \cite{AIK21-acrobatics} (and \Cref{thm:forrelation_AC0_pseudorandom} about the Forrelation distribution $\mathcal{F}_n$), we know that no efficient quantum algorithm that queries the oracle $\mathcal{O}[A]$ can distinguish the distribution of $\{f'_k,g'_k\}_k$ where a random pair $k^*$ is Forrelated from the uniform distribution $\{f'_k,g'_k\}_k$. So, one expects that if one samples $f'_{k^*},g'_{k^*}$ to be uniformly random instead, no quantum algorithm should be able to detect this. However, we cannot use the result of \cite{AIK21-acrobatics} as a black box, because in our setting the quantum algorithm also gets an input state that is correlated with the distribution of the oracle.

To handle this issue, we first show that we can replace the input state with a state that is not correlated with the oracle, so that afterwards we can apply the result of \cite{AIK21-acrobatics}. Ultimately, we will argue using the definition of $\mathcal{F}_n$ that, from the viewpoint of the algorithm, the replaced state looks like a mixture of $\ket{\varphi_{k^*}}$ and a maximally mixed state in the orthogonal subspace.

\comment{M: I think it will make the proof more intuitive if we define the hybrid $H_2$ so that the adversary gets the state $\ket{\Phi_h}$ with probability $\sqrt{\eps}$ and a Haar random state with rest of the probability. Then, the acceptance probabilities under all the hybrids are close to each other and the proof is more intuitive than taking linear combinations of hybrids.

W: not sure I agree, but we can consider it. I plan to leave it as is for QIP, at least

L: I think the issue is that now there is a hybrid that is not defined in a very intuitive way...}

\paragraph{Hybrid $\mathsf{H_2}$:} The distribution of $f_k, g_k, h: \{\pm 1\}^n \to \{\pm 1\}$ is the same as the Hybrid $\mathsf{H_1}$, but the adversary instead receives the state $\ket{\psi} = \ket{\Phi_h}$ as the input (recall that $\ket{\Phi_h} = U_h \ket{+^n}$).\\

Since in Hybrid $\mathsf{H_2}$, the input state is independent of the oracle, we can apply the result of \cite{AIK21-acrobatics} and switch the distribution of $f'_{k^*}, g'_{k^*}$ as discussed before. This gives us the next hybrid.

\paragraph{Hybrid $\mathsf{H_3}$:} The distribution of $f_k, g_k$ is chosen as in the Hybrid $\mathsf{H_2}$ except that $f'_{k^*}, g'_{k^*}$ are chosen to be uniformly random functions as opposed to being sampled from $\mathcal{F}_n$. The adversary receives the same input state $\ket{\psi} = \ket{\Phi_h}$ as in the Hybrid $\mathsf{H_2}$. \\

Note that the distribution of $f_k,g_k$ is uniformly random in $\mathsf{H_3}$, and $\ket{\Phi_h}$ is a random phase state independent of the oracle. The result of \cite{BS19-binary} will imply that we can replace $\ket{\Phi_h}$ with a Haar random state, as in $\mathsf{H_4}$.\footnote{Actually, because we only consider single-copy security, we will not require the full strength of \cite{BS19-binary}. In particular, we will be able to use the simpler observation that a single copy of either a Haar-random state or a random phase state equals the maximally mixed state.}

\subsection{Security Proof}

We now proceed to the formal security proof. For a fixed quantum adversary $\mathcal{A}$ and $i \in \{\mathsf{0}, \mathsf{1}, \mathsf{2}, \mathsf{3}, \mathsf{4}\}$, we denote:
\[
p_i(\mathcal{A}) \coloneqq \Pr_{(\ket{\psi}, A) \sim \mathsf{H}_i}\mbracket{\mathcal{A}^{\mathcal{O}[A]}\mparen{1^\kappa, \ket{\psi}} = 1},
\]
as the probability that the algorithm accepts on a particular hybrid. We successively analyze the hybrids in numerical order.

\begin{claim}
\label{claim:h0_h1}
For all $\mathcal{A}$, $p_\mathsf{0}(\mathcal{A}) = p_\mathsf{1}(\mathcal{A})$.
\end{claim}

\begin{proof}
This follows from the fact that $\mathsf{H_0}$ and $\mathsf{H_1}$ are identically distributed, as we now argue. It suffices to show that the oracle $A$ chosen in $\mathsf{H_1}$ is uniformly random. This holds by observing that if we sample $(f, g) \sim \mathcal{F}_n$ and $h : \{\pm 1\}^n \to \{\pm 1\}$ uniformly at random, then $(f, g \cdot h)$ is a uniformly random pair of functions, because by \Cref{def:forrelation}, the marginal distribution of $f$ is uniformly random.
\end{proof}



For the next pair of hybrids, it will be helpful to first establish two statistical lemmas about the Fourier spectrum of the $f_k$'s.

\begin{lemma}
\label{lem:fourier_bounded_prob_1}
    With probability $1$ over $A$, for all sufficiently large $\kappa$, $k \in \{0,1\}^\kappa$, and $i \in \{0,1\}^n$, we have that:
    \[
        \left| \hat{f}_k(i)\right| \le \frac{1}{\sqrt{\eps 2^n}},
    \]
    where $\eps$ is given in \Cref{def:forrelation}.
\end{lemma}

\begin{proof}
Fix $\kappa \in \Naturals$. Note that for any fixed $k, i$, the Fourier coefficient $\hat{f}_{k}(i)$ is a sum of $2^n$ independent $\pm \frac{1}{2^n}$ random variables. Hence, by \Cref{fact:hoeffding} it holds that:
\begin{align*}
\Pr_{A}\mbracket{\left|\hat{f}_{k}(i)\right| > \frac{1}{\sqrt{\eps 2^n}}}
&\le 2\exp\mparen{-\frac{\frac{2}{\eps 2^n}}{2^n \cdot \frac{4}{4^n}}} =  2\exp\mparen{-\frac{1}{2\eps}},
\end{align*}
and therefore, by a union bound:
\begin{align*}
\Pr_{A}\mbracket{\exists k \in \{0,1\}^\kappa, i \in \{0,1\}^n : \left|\hat{f}_{k}(i)\right| > \frac{1}{ \sqrt{\eps 2^n}}} &\le 2^{n + \kappa + 1}\exp\mparen{-\frac{1}{2\eps}}\\
&\le 2^{n + \kappa + 1 -1/2\eps}\\
&= 2^{\kappa + 1 - 49n}\\
&\le \negl(\kappa),
\end{align*}
where we have used the fact that $\eps = \frac{1}{100n}$ and $n > \kappa$. \comment{L: In this step, it is invoked that $n \ge \kappa/48$, which could be relevant for the multi-copy + output-shrinking setting.}

By the Borel--Cantelli Lemma, because $\sum_{\kappa = 1}^\infty \negl(\kappa) \le O(1)$, we conclude that with probability $1$ over $A$, $\left|\hat{f}_{k}(i)\right| > \frac{1}{\sqrt{\eps 2^n}}$ for at most finitely many $k, i$. Hence, the lemma.
\end{proof}

\begin{lemma}
\label{lem:fourier_f_squared_uniform_prob_1}
    With probability $1$ over $A$,
    \[
         \sum_{i \in \{\pm 1\}^n}\abs{\frac1{2^n} - \E_{k \sim \{0, 1\}^\kappa}\left[\hat f_k(i)^2\right]} \le \negl(\kappa).
    \]
\end{lemma}

\begin{proof}
For notational simplicity, let $q_{\kappa, A} = \sum_{i \in \{\pm 1\}^n}\abs{\frac1{2^n} - \E_{k \sim \{0, 1\}^\kappa}\left[\hat f_k(i)^2\right]}$. Applying \Cref{lem:fourier-concentration}, we have that for an absolute constant $C> 0$,
\[
\E_A\mbracket{q_{\kappa, A}}
\le \sum_{i \in \{\pm 1\}^n} C 2^{-n - \kappa / 2} = C 2^{-\kappa / 2}.
\]
Hence, by Markov's inequality:
\[
    \Pr_A\mbracket{ q_{\kappa, A} \ge \sqrt{C} 2^{-\kappa / 4}} \le \sqrt{C} 2^{-\kappa / 4}.
\]

By the Borel--Cantelli Lemma, because $\sum_{\kappa = 1}^\infty \sqrt{C} 2^{-\kappa / 4} \le O(1)$, we conclude that with probability $1$ over $A$, $q_{\kappa, A} \ge \sqrt{C} 2^{-\kappa / 4}$ for at most finitely many $\kappa \in \Naturals$. This is to say that $q_{\kappa, A} \le \negl(\kappa)$ with probability $1$ over $A$.
\end{proof}

For a given oracle $A$, let $\rho_A$ denote the mixed state obtained by conditionally averaging over all possible states $\ket{\psi}$ such that $(\ket{\psi}, A)$ was sampled from $\mathsf{H_1}$, i.e.\ the 2-Forrelation state $\ket{\varphi_{k^*}}$. That is, we define:
\begin{equation}
\label{eq:def_rho_A}
\rho_A \coloneqq \E_{\mathsf{H_1}} \mbracket{\ket{\psi}\bra{\psi} \mid A}.
\end{equation}

Likewise, define $\sigma_A$ analogously for $\mathsf{H_2}$, i.e.\ the phase state $\ket{\Phi_h}$:
\begin{equation}
\label{eq:def_sigma_A}
\sigma_A \coloneqq \E_{\mathsf{H_2}} \mbracket{\ket{\psi}\bra{\psi} \mid A}.
\end{equation}

Note that the above mixed states are the input states from the viewpoint of any algorithm $\mathcal{A}$ that operates on hybrids $\mathsf{H_1}$ and $\mathsf{H_2}$, respectively, after fixing the oracle $A$.

\begin{lemma}
\label{lem:tau_rho_sigma}
Let $\tau_A = \eps \rho_A + \left(1 - \eps\right)\frac{I}{2^n}$, where $\eps = \frac{1}{100n}$ is as in \Cref{def:forrelation}. Then with probability $1$ over $A$, $\TD(\sigma_A, \tau_A) \le \negl(\kappa)$.
\end{lemma}

\begin{proof}
First, it will be convenient to compute more explicit forms for $\rho_A$ and $\sigma_A$. Letting $\{(f_k, g_k)\}_{k \in \{0,1\}^\kappa}$ be the functions sampled in $A$, we can write:
\[
\rho_A = \E_{k^* \sim \{0,1\}^\kappa}\mbracket{U_{g_{k^*}} H U_{f_{k^*}} \ket{+^n}\bra{+^n} U_{f_{k^*}} H U_{g_{k^*}}}.
\]
Hence, it follows that individual entries of $\rho_A$ are given by:
\begin{equation}
\label{eq:rho_ij}
\bra{i} \rho_A \ket{j} = \E_{k^* \sim \{0,1\}^\kappa} \mbracket{g_{k^*}(i)g_{k^*}(j)\hat{f}_{k^*}(i)\hat{f}_{k^*}(j)},
\end{equation}
where we have used the fact that $\bra{i}H U_f \ket{+^n} = \hat{f}(i)$ for any Boolean function $f$.


Analogously, $\sigma_A$ may be expressed as:
\[
\sigma_A = \E_{k^* \sim \{0,1\}^\kappa} \mbracket{\E_{\mathcal{F}_n}\mbracket{U_{g_{k^*}} U_{g} \ket{+^n}\bra{+^n}U_{g} U_{g_{k^*}} \middle| f = f_{k^*}}},
\]
where the inner expectation denotes that we conditionally average over $(f, g) \sim \mathcal{F}_n$ conditioned on the event $f = f_{k^*}$.
This is identically distributed as $\sigma_A$ since $h = g_{k^*} \cdot g$.
It follows that the entries of $\sigma_A$ are:
\begin{equation}
\label{eq:sigma_ij}
\bra{i} \sigma_A \ket{j} = \E_{k^* \sim \{0,1\}^\kappa} \mbracket{\frac{g_{k^*}(i)g_{k^*}(j) \E_{\mathcal{F}_n}\mbracket{g(i)g(j) \mid f = f_{k^*}}}{2^n}}.
\end{equation}

Our strategy for bounding the expected distance between $\sigma_A$ and $\tau_A$ will be to consider the diagonal and off-diagonal entries separately. 

Fix $i \neq j$. Recall from \Cref{lem:fourier_bounded_prob_1} that with probability $1$ over $A$, for all sufficiently large $\kappa$, for all $k \in \{0,1\}^\kappa$ and $i \in \{0,1\}^n$, $\left|\hat{f}_k(i)\right| \le \frac{1}{\sqrt{\eps 2^n}}$. This implies that $\trnc\mparen{\sqrt{\eps 2^n} \hat{f}(i)} = \sqrt{\eps 2^n} \hat{f}(i)$ and $\trnc\mparen{\sqrt{\eps 2^n} \hat{f}(j)} =  \sqrt{\eps 2^n} \hat{f}(j)$, and therefore:
\begin{equation}
\label{eq:sigma_rho_off_diagonal_proportional}
\E_{\mathcal{F}_n}\mbracket{g(i)g(j) \mid f} = \eps 2^n \hat{f}(i)\hat{f}(j).
\end{equation}
By substituting \eqref{eq:sigma_rho_off_diagonal_proportional} into \eqref{eq:sigma_ij} and comparing with \eqref{eq:rho_ij}, it follows that $\bra{i} \sigma_A \ket{j} = \eps \bra{i} \rho_A \ket{j} = \bra{i} \tau_A \ket{j}$ for every $i \neq j$ (i.e.\ in this case, the off-diagonal entries of $\sigma_A$ and $\tau_A$ are exactly equal). Therefore, with probability $1$ over $A$, for sufficiently large $\kappa$ we have:

\[
\TD(\sigma_A, \tau_A) = \TVD(\diag(\sigma_A), \diag(\tau_A)).
\]
We bound this quantity via:
\begin{align*}
    \TVD(\diag(\sigma_A), \diag(\tau_A))
    &= \TVD(\diag(I/2^n), \diag(\tau_A))\\
    &= \eps \TVD(\diag(I/2^n), \diag(\rho_A))\\
    &= \frac{\eps}{2} \sum_{i \in \{\pm 1\}^n} \left|
\frac{1}{2^n} - \E_{k \sim \{0,1\}^\kappa}\left[ \hat{f}_{k}(i)^2 \right] \right|\\
    &\le \negl(\kappa),
\end{align*}
where in the first line we observe that \eqref{eq:sigma_ij} always evaluates to $\frac{1}{2^n}$ on the diagonal, in the second line we use the fact that $\tau_A$ is a convex combination of $\rho_A$ and $\frac{I}{2^n}$, in the third line we expand the TVD as a sum, and in the last line we appeal to \Cref{lem:fourier_f_squared_uniform_prob_1}, which holds with probability $1$ over $A$.
\end{proof}

\begin{corollary}
\label{cor:h2_mixture}
For all $\mathcal{A}$, $\eps p_\mathsf{1}(\mathcal{A}) + \left(1 - \eps\right) p_\mathsf{3}(\mathcal{A}) - p_\mathsf{2}(\mathcal{A}) \le \negl(\kappa)$.
\end{corollary}

\begin{proof}
By definition of $\rho_A$ \eqref{eq:def_rho_A} and $\sigma_A$ \eqref{eq:def_sigma_A}, we have that:
\[
p_\mathsf{1}(\mathcal{A}) = \Pr_A\mbracket{\mathcal{A}^{\mathcal{O}[A]}\left(1^\kappa, \rho_A\right) = 1}
\]
and 
\[
p_\mathsf{2}(\mathcal{A}) = \Pr_A\mbracket{\mathcal{A}^{\mathcal{O}[A]}\left(1^\kappa, \sigma_A\right) = 1}.
\]
Additionally, observe that
\begin{equation}
\label{eq:p3_I}
p_\mathsf{3}(\mathcal{A}) = \Pr_A\mbracket{\mathcal{A}^{\mathcal{O}[A]}\left(1^\kappa, \frac{I}{2^n}\right) = 1},
\end{equation}
just because, under $\mathsf{H_3}$, $A$ is uniformly random, $h$ is chosen independently of $A$, and
\[
\E_{h : \{\pm 1\}^n \to \{\pm 1\}}\mbracket{\ket{\Phi_h} \bra{\Phi_h}} = \frac{I}{2^n}.
\]
Letting $q$ be the quantity that we wish to bound in the statement of the Corollary, we find that:
\begin{align*}
q &\coloneqq \eps p_\mathsf{1}(\mathcal{A}) + \left(1 - \eps\right) p_\mathsf{3}(\mathcal{A}) - p_\mathsf{2}(\mathcal{A})\\
&= \E_A \mbracket{\Pr\mbracket{\mathcal{A}^{\mathcal{O}[A]}\mparen{1^\kappa, \eps \rho_A + \left(1 - \eps\right) \frac{I}{2^n}} = 1} - \Pr\mbracket{\mathcal{A}^{\mathcal{O}[A]}\mparen{1^\kappa, \sigma_A} = 1}}\\
&= \E_A \mbracket{\Pr\mbracket{\mathcal{A}^{\mathcal{O}[A]}\mparen{1^\kappa, \tau_A} = 1} - \Pr\mbracket{\mathcal{A}^{\mathcal{O}[A]}\mparen{1^\kappa, \sigma_A} = 1}}\\
&\le \E_A\mbracket{\TD(\tau_A, \sigma_A)}\\
&\le \negl(\kappa),
\end{align*}
where in the second line we used the fact that $\mathcal{A}$ is a quantum operation, and hence is linear in its input state; in the third line we substituted the definition of $\tau_A$ in \Cref{lem:tau_rho_sigma}; in the fourth line we appeal to the operational interpretation of trace distance as the maximum bias by which two quantum states can be distinguished; and in the last line we apply \Cref{lem:tau_rho_sigma}.
\end{proof}

This next theorem was essentially proved in \cite[Section 4.2]{AIK21-acrobatics}, with some minor differences in language and choice of parameters. Most of these differences stem from the fact that \cite{AIK21-acrobatics} considered a decision problem called $\textsc{OR} \circ \textsc{Forrelation}$, whereas here we consider a distinguishing task.


\begin{restatable}{theorem}{thmappendix}
\label{thm:h2_h3}
For all polynomial-time quantum adversaries $\mathcal{A}$, $p_\mathsf{2}(\mathcal{A}) - p_\mathsf{3}(\mathcal{A}) \le \negl(\kappa)$.
\end{restatable}

For completeness, we provide proof of \Cref{thm:h2_h3} in \Cref{app:proof_of_h2_h3}.


\begin{corollary}
\label{cor:h1_h2}
For all polynomial-time quantum adversaries $\mathcal{A}$, $p_\mathsf{1}(\mathcal{A}) - p_\mathsf{2}(\mathcal{A}) \le \negl(\kappa)$.
\end{corollary}

\begin{proof}
Recall from \Cref{cor:h2_mixture} that
\[
\eps p_\mathsf{1}(\mathcal{A}) + \left(1 - \eps\right) p_\mathsf{3}(\mathcal{A}) - p_\mathsf{2}(\mathcal{A}) \le \negl(\kappa).
\]
Adding $\left(1 - \eps\right)\left( p_\mathsf{2}(\mathcal{A}) - p_\mathsf{3}(\mathcal{A}) \right)$ to both sides and applying \Cref{thm:h2_h3} gives
\[
\eps \left( p_\mathsf{1}(\mathcal{A}) - p_\mathsf{2}(\mathcal{A}) \right) \le \negl(\kappa).
\]
Multiplying through by $\frac{1}{\eps}$ yields
\[
p_\mathsf{1}(\mathcal{A}) - p_\mathsf{2}(\mathcal{A}) \le \frac{\negl(\kappa)}{\eps} \le \negl(\kappa),
\]
because $\frac{1}{\eps} \le O\left(n\right) \le \poly(\kappa)$.
\end{proof}

Finally, we prove indistinguishability of the remaining two hybrids.

\begin{claim}
\label{claim:h3_h4}
For all $\mathcal{A}$, $p_\mathsf{3}(\mathcal{A}) = p_\mathsf{4}(\mathcal{A})$.
\end{claim}

\begin{proof}
Recall from \eqref{eq:p3_I} in the proof of \Cref{cor:h2_mixture} that
\[
p_\mathsf{3}(\mathcal{A}) = \Pr_A\mbracket{\mathcal{A}^{\mathcal{O}[A]}\mparen{1^\kappa, \frac{I}{2^n}} = 1}.
\]
To complete the proof, notice that
\[
p_\mathsf{4}(\mathcal{A}) = \Pr_{A, \ket{\psi} \sim \mu_{\mathrm{Haar}}^n}\mbracket{\mathcal{A}^{\mathcal{O}[A]}\mparen{1^\kappa, \ket{\psi}} = 1} = \Pr_A\mbracket{\mathcal{A}^{\mathcal{O}[A]}\mparen{1^\kappa, \frac{I}{2^n}} = 1}
\]
as well, which follows from the fact that
\[
\E_{\ket{\psi} \sim \mu_{\mathrm{Haar}}^n}\mbracket{\ket{\psi}\bra{\psi}} = \frac{I}{2^n}.\qedhere
\]
\end{proof}

We can now show that the distinguishing advantage of any efficient adversary is negligible when averaged over the random oracle $A$.

\begin{theorem}
\label{thm:prs_advantage_small}
Let $\mathcal{A}$ be a polynomial-time quantum adversary. Then:
\[
\E_{A}\mbracket{\Pr_{k^* \sim \{0,1\}^\kappa} \mbracket{\mathcal{A}^{\mathcal{O}[A]}\mparen{1^\kappa, \ket{\varphi_{k^*}}} = 1} - \Pr_{\ket{\psi} \sim \mu_{\mathrm{Haar}}^n} \mbracket{\mathcal{A}^{\mathcal{O}[A]}\mparen{1^\kappa, \ket{\psi}} = 1}} \le \negl(\kappa).
\]
\end{theorem}

\begin{proof}
Observe that the quantity that we wish to bound is exactly $p_\mathsf{0}(\mathcal{A}) - p_\mathsf{4}(\mathcal{A})$. From \Cref{claim:h0_h1}, \Cref{cor:h1_h2}, \Cref{thm:h2_h3}, and \Cref{claim:h3_h4}, we know that $p_i(\mathcal{A}) - p_{i+1}(\mathcal{A}) \le \negl(\kappa)$ for every $i \in \{\mathsf{0}, \mathsf{1}, \mathsf{2}, \mathsf{3}\}$. Summing these bounds then gives the desired result.
\end{proof}

Using techniques similar to Yao's distinguisher/predictor lemma \cite{Yao82}, this also yields a bound on the absolute advantage of any adversary. The rough idea is that $\mathcal{A}$ can try to guess the sign of its own distinguishing advantage.

\begin{corollary}
\label{cor:prs_abs_advantage_small}
Let $\mathcal{A}$ be a polynomial-time quantum adversary. Then:
\[
\E_{A} \left|\Pr_{k^* \sim \{0,1\}^\kappa} \mbracket{\mathcal{A}^{\mathcal{O}[A]}\mparen{1^\kappa, \ket{\varphi_{k^*}}} = 1} - \Pr_{\ket{\psi} \sim \mu_{\mathrm{Haar}}^n} \mbracket{\mathcal{A}^{\mathcal{O}[A]}\mparen{1^\kappa, \ket{\psi}} = 1}\right| \le \negl(\kappa).
\]
\end{corollary}

\begin{proof}
We assume that $\mathcal{A}$ outputs a bit in $\{\pm1\}$. For an oracle $A$, let
\[
a(A) \coloneqq \Pr_{k^* \sim \{0,1\}^\kappa} \mbracket{\mathcal{A}^{\mathcal{O}[A]}\mparen{1^\kappa, \ket{\varphi_{k^*}}} = 1}
\]
and
\[
b(A) \coloneqq \Pr_{\ket{\psi} \sim \mu_{\mathrm{Haar}}^n}\mbracket{\mathcal{A}^{\mathcal{O}[A]}\mparen{1^\kappa, \ket{\psi}} = 1}
\]
so that the quantity we need to bound is $\E_{A} \left| a(A) - b(A) \right|$. Consider an adversary $\mathcal{B}\mparen{1^\kappa, \ket{\psi}}$ that does the following:
\begin{enumerate}
\item Toss a coin $c \in \{\pm1\}$.
\item If $c = 1$, then execute $\mathcal{A}\mparen{1^\kappa, \E_{k \sim \{0,1\}^\kappa}\mbracket{\ket{\varphi_k}\bra{\varphi_k}}}$ once and call the output $d$. Otherwise, if $c = -1$, then execute $\mathcal{A}\mparen{1^\kappa, \frac{I}{2^n}}$ and call the output $d$.
\item Output $c \cdot d \cdot \mathcal{A}\mparen{1^\kappa, \ket{\psi}}$.
\end{enumerate}

Observe that $\mathcal{B}$ runs in polynomial time. Also note that $c \cdot d$ is sampled to be $1$ with probability $\frac{1 + a(A) - b(A)}{2}$ and $-1$ with probability $\frac{1 - a(A) + b(A)}{2}$. As a result, we may compute:
\begin{align}
& \E_{A} \left[\Pr_{k^* \sim \{0,1\}^\kappa} \left[\mathcal{B}^{\mathcal{O}[A]}\left(1^\kappa, \ket{\varphi_{k^*}} \right) = 1 \right] - \Pr_{\ket{\psi} \sim \mu_{\mathrm{Haar}}^n} \left[\mathcal{B}^{\mathcal{O}[A]}\left(1^\kappa, \ket{\psi} \right) = 1 \right] \right] \nonumber\\
&\qquad\qquad = \E_{A} \left[\Pr[cd = 1]\left(a(A) - b(A)\right) + \Pr[cd = -1]\left(b(A) - a(A)\right) \right]\nonumber\\
&\qquad\qquad = \E_{A} \left[\left(a(A) - b(A)\right)^2\right].\label{eq:exp_(a-b)^2}
\end{align}

To complete the proof, we bound the quantity:
\begin{align*}
\E_{A} \left| a(A) - b(A) \right| &\le \sqrt{\E_{A} \left[ (a(A) - b(A))^2 \right]}\\
&= \sqrt{\E_{A} \left[\Pr_{k^* \sim \{0,1\}^\kappa} \left[\mathcal{B}^{\mathcal{O}[A]}\left(1^\kappa, \ket{\varphi_{k^*}} \right) = 1 \right] - \Pr_{\ket{\psi} \sim \mu_{\mathrm{Haar}}^n} \left[\mathcal{B}^{\mathcal{O}[A]}\left(1^\kappa, \ket{\psi} \right) = 1 \right] \right]}\\
&\le \sqrt{\negl(\kappa)}\\
&\le \negl(\kappa),
\end{align*}
where in the first line we applied Jensen's inequality, in the second line we substituted \eqref{eq:exp_(a-b)^2}, and in the third line we applied \Cref{thm:prs_advantage_small}.
\end{proof}

Finally, we complete the proof of \Cref{thm:single_copy_prs_prob_1} to show that $\{\ket{\varphi_k}\}_{k \in \{0,1\}^\kappa}$ is pseudorandom.

\begin{proof}[Proof of \Cref{thm:single_copy_prs_prob_1}]
It is clear that the ensemble satisfies the efficient generation criterion of \Cref{def:single_copy_prs}. Indeed, for any $k \in \{0,1\}^\kappa$, $U_{f_k}$ and $U_{g_k}$ can each be implemented using a single query to $\mathcal{A}$ on an input of length $n + \kappa + 1 \le \poly(\kappa)$, and hence $\ket{\varphi_k}$ can be prepared in polynomial time.

It remains only to establish computational indisginguishability from the Haar measure. By \Cref{cor:prs_abs_advantage_small}, for every polynomial-time quantum adversary $\mathcal{A}$, there exists a negligible function $\delta(\kappa)$ for which:
\[
\E_{A} \left|\Pr_{k^* \sim \{0,1\}^\kappa} \left[\mathcal{A}^{\mathcal{O}[A]}\left(1^\kappa, \ket{\varphi_{k^*}} \right) = 1 \right] - \Pr_{\ket{\psi} \sim \mu_{\mathrm{Haar}}^n} \left[\mathcal{A}^{\mathcal{O}[A]}\left(1^\kappa, \ket{\psi} \right) = 1 \right] \right| \le \delta(\kappa).
\]
Hence, by Markov's inequality, we know that:
\[
\Pr_{A} \left[\left|\Pr_{k^* \sim \{0,1\}^\kappa} \left[\mathcal{A}^{\mathcal{O}[A]}\left(1^\kappa, \ket{\varphi_{k^*}} \right) = 1 \right] - \Pr_{\ket{\psi} \sim \mu_{\mathrm{Haar}}^n} \left[\mathcal{A}^{\mathcal{O}[A]}\left(1^\kappa, \ket{\psi} \right) = 1 \right] \right| \ge \sqrt{\delta(\kappa)}\right] \le \sqrt{\delta(\kappa)}.
\]
Since $\delta(\kappa)$ is a negligible function, the infinite sum $\sum_{\kappa = 0}^\infty \sqrt{\delta(\kappa)}$ must be bounded by some constant. Therefore, by the Borel--Cantelli Lemma, with probability $1$ over $A$ we have
\[
\left|\Pr_{k^* \sim \{0,1\}^\kappa} \left[\mathcal{A}^{\mathcal{O}[A]}\left(1^\kappa, \ket{\varphi_{k^*}} \right) = 1 \right] - \Pr_{\ket{\psi} \sim \mu_{\mathrm{Haar}}^n} \left[\mathcal{A}^{\mathcal{O}[A]}\left(1^\kappa, \ket{\psi} \right) = 1 \right] \right| \le \sqrt{\delta(\kappa)}
\]
for all but at most finitely many $\kappa$. This is to say that with probability $1$ over $A$,
\begin{equation}
\label{eq:A_advantage_negligible}
\left|\Pr_{k^* \sim \{0,1\}^\kappa} \left[\mathcal{A}^{\mathcal{O}[A]}\left(1^\kappa, \ket{\varphi_{k^*}} \right) = 1 \right] - \Pr_{\ket{\psi} \sim \mu_{\mathrm{Haar}}^n} \left[\mathcal{A}^{\mathcal{O}[A]}\left(1^\kappa, \ket{\psi} \right) = 1 \right] \right| \le \negl(\kappa).
\end{equation}
Since there are only countably many uniform quantum adversaries $\mathcal{A}$, a union bound over all such $\mathcal{A}$ implies that with probability $1$ over $A$, \textit{every} polynomial-time quantum adversary $\mathcal{A}$ satisfies \eqref{eq:A_advantage_negligible}. This is to say that the state ensemble satisfies \Cref{def:single_copy_prs}.
\end{proof}

\section{Implications for the Standard Model}
\label{sec:standard_model}

We now make a few remarks about how the security proof above can be ported to the nonoracular setting. In particular, we argue that our security proof gives a way to instantiate real-world pseudorandom states without assuming the existence of one-way functions. We do so by considering the following security property of a (nonoracular) set of functions $F = \{(f_k, g_k)\}_{k \in \{0, 1\}^\kappa}$:

\begin{property}
    \label{property}
    Let $F = \{(f_k, g_k)\}_{k \in \{0,1\}^\kappa}$ be a set of pairs of functions $f_k, g_k: \{\pm 1\}^n \to \{\pm 1\}$ keyed by $k$, for some $\kappa + 1 \le n \le \poly(\kappa)$. We assume $F$ satisfies the following:
    \begin{enumerate}[(i)]
  \item \label{item:property_efficient_gen} (Efficient computation) For all $k$, $f_k$ and $g_k$ can be evaluated in time $\poly(\kappa)$.
  \item \label{item:fourier_smoothness} (Smoothness of Fourier spectrum of $f_k$) For all sufficiently large $\kappa$, $k \in \{0, 1\}^\kappa$, and $i \in \{0, 1\}^n$, we have that:
  \begin{equation}
    \label{eq:fourier_f_bounded}
      \left|\hat f_k(i)\right| \le \frac1{\sqrt{\eps 2^n}},
  \end{equation}
  where $\eps$ is given in \Cref{def:forrelation}. 
  Additionally, we have:
  \begin{equation}
  \label{eq:fourier_f_squared_uniform}
      \sum_{i \in \{\pm 1\}^n}\abs{\frac1{2^n} - \E_{k \in \{0, 1\}^\kappa}\left[\hat f_k(i)^2\right]} \le \negl(\kappa).
  \end{equation}
  \item \label{item:shifted_forrelation} (Hardness of shifted Forrelation)
    Let $h \sim \mathcal{H}_\kappa$ denote that we sample $h$ as follows. First, we choose $k \sim \{0,1\}^\kappa$. Then, we sample a function $g: \{\pm 1\}^n \to \{\pm 1\}$ by sampling $g(x)$ with bias $\sqrt{\eps 2^n} \hat{f}_k(x)$, independently for each $x \in \{\pm 1\}^n$. Equivalently, we sample from the conditional probability distribution $g \sim \mathcal{F}_n \mid f = f_k$. Finally, we let $h = g_k \cdot g$.
    
    For any polynomial-time quantum adversary $\mathcal{A}$ with quantum query access to $h$, we require that:
    \begin{equation}
    \left|
    \Pr_{h \sim \mathcal{H}_\kappa}\left[\mathcal{A}^h\left(1^\kappa \right) = 1 \right]
    -
    \Pr_{h : \{\pm 1\}^n \to \{\pm 1\}}\left[\mathcal{A}^h\left(1^\kappa \right) = 1 \right]
    \right| \le \negl(\kappa).
    \end{equation}
\end{enumerate}
\end{property}
The high level intuition of \eqref{item:shifted_forrelation} is that an efficient quantum algorithm on input a shift function $h$ (via oracle access) cannot approximate the maximum shifted Forrelation value: $\max_{k \in \{0, 1\}^\kappa} \abs{\braket{+^n|\Phi_{(f_k, g_k \cdot h)}}}$; or equivalently, the algorithm is not able to distinguish a uniformly random $h$ from $h$ such that for some $k$, $f_k$ and $g_k \cdot h$ are noticeably Forrelated.

\subsection{Usefulness of \texorpdfstring{\Cref{property}}{Property \ref{property}}}

In light of our proofs in \Cref{sec:single_security}, we observe that \Cref{property} simultaneously:
\begin{enumerate}[(a)]
    \item \label{item:sufficient_for_prs} Suffices to construct $n$-qubit single-copy pseudorandom states,
    \item \label{item:holds_random} Holds for a random oracle, and thus plausibly holds for existing cryptographic hash functions like SHA-3, and
    \item \label{item:independent_p_np} Is independent of $\mathsf{P}$ vs.\ $\mathsf{NP}$ in the black-box setting. \comment{L: Is it fair to say that this holds for a random oracle?}
\end{enumerate}

We briefly sketch why this is the case. To establish \eqref{item:sufficient_for_prs}, we note that the same general hybrid argument suffices to establish the single-copy pseudorandomness of the ensemble $\{\ket{\varphi_k} \coloneqq \ket{\Phi_{(f_k, g_k)}}\}_{k \in \{0,1\}^\kappa}$, assuming $F$ satisfies \Cref{property}. Consider a sequence of hybrids where in $\mathsf{H_1}$, the adversary receives $\ket{\varphi_k}$ for a random k; in $\mathsf{H_2}$, the adversary receives $\ket{\Phi_h}$ for $h \sim \mathcal{H}_\kappa$; in $\mathsf{H_3}$, the adversary receives $\ket{\Phi_h}$ for a uniformly random $h$; and in $\mathsf{H_4}$, the adversary receives a Haar-random state $\ket{\psi}$. Then, \cref{item:fourier_smoothness} of \Cref{property} serves as a substitute for \Cref{lem:fourier_bounded_prob_1} and \Cref{lem:fourier_f_squared_uniform_prob_1}, and implies via the same argument as \Cref{lem:tau_rho_sigma} that $\mathsf{H_2}$ is statistically indistinguishable from a non-negligible mixture of $\mathsf{H_1}$ and $\mathsf{H_3}$. \cref{item:shifted_forrelation} of \Cref{property} implies that $\mathsf{H_2}$ and $\mathsf{H_3}$ are computationally indistinguishable, because $\ket{\Phi_h}$ can be prepared efficiently with a single query to $h$. Finally, $\mathsf{H_3}$ and $\mathsf{H_4}$ are statistically indistinguishable, just because
\[
\E_{\ket{\psi} \sim \mu_{\mathrm{Haar}}^n}\mbracket{\ket{\psi}\bra{\psi}} = \E_{h : \{\pm 1\}^n \to \{\pm 1\}}\mbracket{\ket{\Phi_h}\bra{\Phi_h}} = \frac{I}{2^n},
\]
as observed in \Cref{claim:h3_h4}. Together, these imply that $\mathsf{H_1}$ and $\mathsf{H_4}$ are computationally indistinguishable, which proves the pseudorandomness of the ensemble.

On the other hand, \eqref{item:holds_random} and \eqref{item:independent_p_np} can be established simultaneously by showing that \Cref{property} holds relative to $\mathcal{O}[A]$, with probability $1$ over a random oracle $A$. This is because $A$ is a sub-oracle of $\mathcal{O}[A]$, so if \Cref{property} holds relative to $\mathcal{O}[A]$, then it certainly holds relative to $A$ alone, because the functions $f_k$ and $g_k$ depend only on $A$. Additionally, we know that $\mathsf{P}^A \neq \mathsf{NP}^A$ with probability $1$ over $A$ \cite{BG81-random-oracle}, whereas $\mathsf{P}^{\mathcal{O}[A]} = \mathsf{NP}^{\mathcal{O}[A]}$, by \Cref{prop:oracle_p_equals_np}.

\cref{item:property_efficient_gen} clearly holds relative to $\mathcal{O}[A]$, by \Cref{def:prs_oracle_ensemble}, whereas \eqref{item:fourier_smoothness} was shown to hold with probability $1$ in \Cref{lem:fourier_bounded_prob_1} and \Cref{lem:fourier_f_squared_uniform_prob_1}. Finally, \eqref{item:shifted_forrelation} was essentially established in the proof of \Cref{thm:h2_h3} in \Cref{app:proof_of_h2_h3}. Technically, one caveat is that our proof only showed:
\begin{equation}
\label{eq:oracle_caveat_1}
\E_A \mbracket{\Pr_{h \sim \mathcal{H}_\kappa}\left[\mathcal{A}^{\mathcal{O}[A], h}\left(1^\kappa \right) = 1 \right]
    -
    \Pr_{h : \{\pm 1\}^n \to \{\pm 1\}}\left[\mathcal{A}^{\mathcal{O}[A], h}\left(1^\kappa \right) = 1 \right]
     } \le \negl(\kappa),
\end{equation}
and this is not the same as
\begin{equation}
\label{eq:oracle_caveat_2}
\E_A\abs{\Pr_{h \sim \mathcal{H}_\kappa}\left[\mathcal{A}^{\mathcal{O}[A], h}\left(1^\kappa \right) = 1 \right]
    -
    \Pr_{h : \{\pm 1\}^n \to \{\pm 1\}}\left[\mathcal{A}^{\mathcal{O}[A], h}\left(1^\kappa \right) = 1 \right]
     } \le \negl(\kappa).
\end{equation}
The latter is needed to conclude, by the Borel--Cantelli lemma, that \eqref{item:shifted_forrelation} holds with probability $1$ relative to $\mathcal{O}[A]$. We believe that a minor modification of our proof strategy would yield this stronger claim.\footnote{Indeed, it would suffice to prove that \eqref{eq:oracle_caveat_1} holds for adversaries $\mathcal{A}$ with a single bit of advice in order to conclude \eqref{eq:oracle_caveat_2}, since in that case we can assume without loss of generality that the distinguishing advantage of $\mathcal{A}$ is nonnegative. Alternatively, it might be possible to prove an analogue of \Cref{cor:prs_abs_advantage_small} by having $\mathcal{A}$ guess the direction of its own distinguishing advantage. The difficulty is that there does not seem to be a way to efficiently simulate queries to the security challenge $h \sim \mathcal{H}_\kappa$, as we discuss further below.
}

\subsection{Further Remarks}


A few additional comments are in order. First, we emphasize that \eqref{item:shifted_forrelation} is the only computational hardness property assumed in \Cref{property}. Indeed, \eqref{item:fourier_smoothness} is merely a \textit{statistical} property of the functions $f_k$. Thus, we could gain confidence that \eqref{item:fourier_smoothness} holds for a specific $F$ by verifying on small values of $\kappa$, or we might even be able to prove that it holds unconditionally for certain instantiations of $F$. Furthermore, this statistical property as stated is sufficient but perhaps not necessary for our proofs to go through. For instance, we believe that one could relax $\eqref{eq:fourier_f_bounded}$ to only hold with overwhelming probability over uniformly chosen $k \in \{0, 1\}^\kappa$.

One might object that the security property \eqref{item:shifted_forrelation} is impractical and unrealistic, because there is no way to efficiently simulate quantum query access to a random $h \sim \mathcal{H}_\kappa$. In the language of Naor \cite{Naor03-assumption} and Gentry and Wichs \cite{GentryW11-falsifiable}, \Cref{property} is not \textit{falsifiable}, because the security property cannot (apparently) be modeled as an interactive game between an adversary and an efficient challenger, in which the challenger can decide whether the adversary won the game. However, we emphasize that efficient simulation is not actually necessary for a security property to be useful! Indeed, it is quite common for cryptographic security reductions to proceed via a hybrid argument in which one or more of the hybrids has no efficient simulation, as we have done here. We also note that exactly the same criticism could be leveled against the definition of pseudorandom states itself, because Haar-random quantum states cannot be prepared in polynomial time. And yet, we know that pseudorandom states \textit{are} useful for instantiating a wide variety of cryptographic schemes \cite{JLS18-prs,AQY22-prs,MY22-prs}.

\section{Conjectured Multi-Copy Security}

\label{sec:conjectured_multi_security}

In this section, we outline a plausible path towards proving that our oracle construction remains secure in the multi-copy case, assuming a strong conjecture about $t$-Forrelation states. To motivate this conjecture, it will be helpful to identify the step in our proof of single-copy security that breaks down in the multi-copy case. The key step appears to be \Cref{lem:tau_rho_sigma}, which essentially states that the view of the adversary under $\mathsf{H_2}$ is equivalent to a probabilistic mixture of its views under $\mathsf{H_1}$ and $\mathsf{H_3}$. This relies on the fact that, for a given $A$ and $k^*$, the state $\ket{\Phi_h}$ sampled in $\mathsf{H_2}$ will have $\E_{\mathsf{H_2}}\left[\left|\braket{\Phi_h | \varphi_{k^*}} \right| \mid A, k^* \right] = \delta$ for some non-negligible $\delta$. Unfortunately, this does not appear to hold in the multi-copy case: the expected overlap between $T$ copies of the states $\E_{\mathsf{H_2}} \left[\left|\bra{\Phi_h}^{\otimes T}\ket{\varphi_{k^*}}^{\otimes T} \right| \mid A, k^* \right]$ could be much smaller, typically on the order of $\delta^T$, which can quickly become negligible for $T = \poly(\kappa)$. Thus, the correlation between $\ket{\Phi_h}$ and $\ket{\varphi_{k^*}}$ is too small to directly prove indistinguishability in the multi-copy case. Note, however, that we have \textit{not} demonstrated multi-copy insecurity of our construction from \Cref{sec:single_security}. Rather, it just appears that proving multi-copy security would require different ideas.

\comment{M: Actually I think the current construction should remain secure for $T \approx \kappa$ copies since the distinguishing advantage of $\text{OR} \circ \text{Forrelation}$ is $2^{-\kappa}$ instead of $\negl(\kappa)$ so we can tolerate $T$ satisfying $\delta^T \approx 2^{-\kappa}$. We should be able to prove this just by some concentration arguments hopefully, although probably not before the QIP deadline. 

Also, I think a construction based on $\text{OR} \circ ($t$\text{-fold Forrelation})$ should allow us to improve it to any polynomial number of copies since the distinguishing advantage for $t$-fold Forrelation should be roughly $2^{-t\kappa}$. The only reason this does not work with the current proof is because in the current hard distribution for t-fold Forrelation $\delta = 2^{-t}$ (also decreases with $t$), but if we could improve it slightly to say $\delta = 2^{-t^{c}}$ for $c<1$, then we should be able to trade off the parameters. This is a much weaker conjecture than the one below which requires $\delta = 1-\negl(t)$.
}

\subsection{Our Conjecture}

To overcome this issue, we conjecture the existence of ``Forrelation-like'' distributions with much stronger correlation properties. The formal statement of our conjecture is the following:


\begin{conjecture}
\label{conj:strong_forrelation}
For some $t = \poly(n)$, for every $t$-tuple $G = (g^1, g^2, ..., g^t)$ where $g^i : \{\pm 1\}^n \to \{\pm 1 \}$, there exists a distribution $\mathcal{D}_G$ over $t$-tuples of functions $F = (f^1, f^2, ..., f^t)$ where $f^i : \{\pm 1\}^n \to \{\pm 1 \}$ such that:
\begin{enumerate}[(i)]
\item \label{item:conj_pseudorandom} (Pseudorandom against $\mathsf{AC^0}$) For every $C \in \mathsf{AC^0}[2^{\poly(n)}, O(1)]$,
\[
\left|\E_{F \sim \mathcal{D}_G}\left[C\left(\ttable(f^1), \ttable(f^2), \ldots, \ttable(f^t) \right)\right] - \E_{z \sim \{\pm 1\}^{t2^n}}\left[C(z)\right]\right| \le \negl(n).
\]
\item \label{item:conj_closeness} (Statistical closeness to $\ket{\Phi_G}$)
\[
\E_{F \sim \mathcal{D}_G} \left[\TD(\ket{\Phi_F}, \ket{\Phi_G}) \right] \le \mathrm{negl}(n).
\]
\item \label{item:conj_compound_uniform} (Compound distribution is uniform) If $G$ is a uniformly random $t$-tuple of functions, then sampling $F \sim \mathcal{D}_G$ yields a uniformly random $t$-tuple of functions (averaged over $G$).
\end{enumerate}
\end{conjecture}

To provide some intuition, we state a weaker conjecture that is implied by \Cref{conj:strong_forrelation}, and that is more directly comparable to the currently known properties of the Forrelation distribution. Note, however, that we only know how to prove multi-copy security assuming the stronger \Cref{conj:strong_forrelation}; it is not clear whether the weaker conjecture below suffices.

\begin{conjecture}
\label{conj:weak_forrelation}
For some $t = \poly(n)$, there exists a distribution $\mathcal{D}$ over $t$-tuples of functions $F = (f^1, f^2, ..., f^t)$ where $f^i : \{\pm 1\}^n \to \{\pm 1 \}$ such that:
\begin{enumerate}[(i)]
\item \label{item:conj_pseudorandom_weak} (Pseudorandom against $\mathsf{AC^0}$) For every $C \in \mathsf{AC^0}[2^{\poly(n)}, O(1)]$,
\[
\left|\E_{F \sim \mathcal{D}}\left[C\left(\ttable(f^1), \ttable(f^2), \ldots, \ttable(f^t) \right)\right] - \E_{z \sim \{\pm 1\}^{t2^n}}\left[C(z)\right]\right| \le \negl(n).
\]
\item \label{item:conj_closeness_weak} (Statistical closeness to $\ket{+^n}$)
\[
\E_{F \sim \mathcal{D}} \left[\braket{+^n | \Phi_F} \right] \ge 1 - \mathrm{negl}(n).
\]
\end{enumerate}
\end{conjecture}

In plain words, \Cref{conj:weak_forrelation} posits the existence of a distribution that is pseudorandom against $\mathsf{AC^0}$, and that samples \textit{highly $t$-Forrelated} functions with high probability. If we compare to what is known about the $t=2$ case, we know by \Cref{thm:forrelation_AC0_pseudorandom} that the Forrelation distribution $\mathcal{F}_n$ is also pseudorandom against $\mathsf{AC^0}$. However, $\mathcal{F}_n$ only samples functions that are \textit{weakly Forrelated}, i.e.\ $\E_{F \sim \mathcal{F}_n} \left[\braket{+^n | \Phi_F} \right] = \delta$ for some non-negligible $\delta$, rather than a $\delta$ that is close to $1$. For values of $t > 2$, the current state of the art for $t$-fold Forrelation~\cite{BS21} gives a distribution over $t$-tuples of functions that is pseudorandom against $\mathsf{AC^0}$ circuits (in fact, the pseudorandomness parameter is $2^{-\Omega(nt)}$ as opposed to $\negl(n)$), however the expected overlap of $\ket{\Phi_F}$ with the $\ket{+^n}$ state is roughly $2^{-\Omega(t)}$ which is not sufficient for our purposes.

We expect that it may be necessary to choose $t$ to be some large polynomial, say $t = n^2$, in order for either \Cref{conj:strong_forrelation} or \Cref{conj:weak_forrelation} to hold. The reason is that, for small $t$, there could be very few $F$s for which $\ket{\Phi_F}$ has large overlap with the $\ket{+^n}$ state, and so it might not be possible for a distribution over such $F$s to also be pseudorandom against $\mathsf{AC^0}$. In more technical terms, a counting argument suggests that random $2$-Forrelation states would not form an $\eps$-net to the set of $n$-qubit states with real amplitudes, at least not for small $\eps$. However, for larger $t$, it seems plausible that $t$-Forrelation states could form an $\eps$-net for some exponentially small $\eps$, and proving this might be a useful first step towards establishing either of our two conjectures.


\subsection{Security Proof Sketch}

Throughout this section, we assume \Cref{conj:strong_forrelation} holds for some fixed $t(n)$, and take this value of $t$ in \Cref{def:prs_oracle_ensemble}. Like in \Cref{sec:single_security}, our conditional security proof proceeds via a hybrid argument. As in \Cref{sec:single_security}, each hybrid defines a security challenge for the quantum adversary consisting of a state $\ket{\psi}$ and an oracle $A$, and the adversary makes queries to $\mathcal{O}[A]$. Unlike in the single-copy case, the adversary is given an arbitrary polynomial number of copies of $\ket{\psi}$ (i.e.\ $\ket{\psi}^{\otimes T}$ for any $T = \poly(\kappa)$ chosen by the adversary).

For convenience, in each of these hybrids we only specify the part of $A$ that corresponds to the functions $\{(f_k^1, f_k^2, \ldots, f_k^t)\}_{k \in \{0,1\}^\kappa}$ that are used to construct the states with security parameter $\kappa$. Recall that $f_k^i(x) =  A(x, k, i)$. Otherwise, the rest of $A$ is always sampled uniformly at random.

\paragraph{Hybrid $\mathsf{H_0}$:} Sample $k^* \sim \{0,1\}^\kappa$. For each $k \in \{0,1\}^\kappa$, sample $F_k$ uniformly at random. The adversary gets $\ket{\psi} = \ket{\varphi_{k^*}}$ as input.

\paragraph{Hybrid $\mathsf{H_1}$:}
Sample $G = (g^1, g^2, \ldots, g^t)$ uniformly at random. Sample $k^* \sim \{0,1\}^\kappa$. For each $k \in \{0,1\}^\kappa$, sample $F_k$ as follows:
\begin{itemize}
\item If $k = k^*$, draw $F_k \sim \mathcal{D}_{G}$.
\item If $k \neq k^*$, draw $F_k$ uniformly at random.
\end{itemize}
The adversary gets $\ket{\psi} = \ket{\Phi_{F_{k^*}}}$ as input.

\paragraph{Hybrid $\mathsf{H_2}$:} The same as Hybrid $\mathsf{H_1}$, but the adversary instead receives the state $\ket{\psi} = \ket{\Phi_{G}}$.

\paragraph{Hybrid $\mathsf{H_3}$:} The same as Hybrid $\mathsf{H_2}$, but now we draw $F_{k^*}$ uniformly at random.

\paragraph{Hybrid $\mathsf{H_4}$:} For each $k \in \{0,1\}^\kappa$, sample $F_k$ uniformly at random. The adversary gets a Haar-random state $\ket{\psi}$ as input.\\

Assuming \Cref{conj:strong_forrelation}, the rough idea of the security proof is largely the same as in the single-copy case, so we provide only a brief sketch of security.

\begin{theorem}
Assuming \Cref{conj:strong_forrelation}, with probability $1$ over $A$, $\{\ket{\varphi_k}\}_{k \in \{0,1\}^\kappa}$ is multi-copy pseudorandom relative to $\mathcal{O}[A]$.
\end{theorem}

\begin{proof}[Proof sketch]
    Hybrids $\mathsf{H_0}$ and $\mathsf{H_1}$ are indistinguishable (indeed, statistically identical from the view of the adversary) by \cref{item:conj_compound_uniform} of \Cref{conj:strong_forrelation}.
    
    Hybrids $\mathsf{H_1}$ and $\mathsf{H_2}$ are indistinguishable because for any $T = \poly(\kappa)$,
    \begin{align*}
        \E_{F \sim \mathcal{D}_G}\left[\TD\left(\ket{\Phi_F}^{\otimes T}, \ket{\Phi_G}^{\otimes T} \right) \right] &\le T \cdot \E_{F \sim \mathcal{D}_G}\left[\TD\left(\ket{\Phi_F}, \ket{\Phi_G} \right) \right]\\
        &\le T \cdot \negl(n)\\
        &\le \negl(\kappa),
    \end{align*}
    where the first line holds by the subadditivity of trace distance under tensor products, the second line holds by \cref{item:conj_closeness} of \Cref{conj:strong_forrelation}, and the third line holds because $n > \kappa$.

    Hybrids $\mathsf{H_2}$ and $\mathsf{H_3}$ are indistinguishable via the same argument as \Cref{thm:h2_h3} (proved in \Cref{app:proof_of_h2_h3}), by replacing the distribution $\mathcal{F}_n$ with $\mathcal{D}_G$. Indeed, the only property of $\mathcal{F}_n$ that is needed to prove \Cref{thm:h2_h3} is its pseudorandomness against $\mathsf{AC^0}$ (\Cref{thm:forrelation_AC0_pseudorandom}), and \cref{item:conj_pseudorandom} of \Cref{conj:strong_forrelation} posits that a similar pseudorandomness property holds for $\mathcal{D}_G$.

    Hybrids $\mathsf{H_3}$ and $\mathsf{H_4}$ are indistinguishable by the result of Brakerski and Shmueli \cite{BS19-binary}. In more detail, we know that in $\mathsf{H_3}$ the adversary always receives a random $t$-Forrelation state $\ket{\Phi_G}$, whereas in $\mathsf{H_4}$ the adversary always receives a Haar-random state $\ket{\psi}$. Hence, for any $T = \poly(\kappa)$, the distinguishing advantage of the adversary with $T$ copies of the state is upper bounded by:

    \begin{align*}
        \TD \left(\E_{G} \left[\ket{\Phi_G}\bra{\Phi_G}^{\otimes T} \right], \E_{\ket{\psi} \sim \mu_{\mathrm{Haar}}} \left[\ket{\psi}\bra{\psi}^{\otimes T} \right] \right) &\le \negl(n)\\
        &\le \negl(\kappa),
    \end{align*}
    where the first line was shown in \cite{BS19-binary}\footnote{Strictly speaking, \cite{BS19-binary} only showed this for the case in which $\ket{\Phi_G}$ is a random $1$-Forrelation state (i.e.\ a random binary phase state), not a $t$-Forrelation state. However, a random $t$-Forrelation state can only be ``more random'': indeed, it can be obtained from a random $1$-Forrelation state via multiplication by an independent random unitary of the form $U_{g^t} \cdot H \cdots U_{g^2} \cdot H$. Thus, by the unitary invariance of the Haar measure, the generalization to $t$-Forrelation states also holds.} and the second line holds because $n > \kappa$.

    Putting these bounds together, we have that (c.f. \Cref{thm:prs_advantage_small}), for any polynomial-time quantum adversary $\mathcal{A}$:
\[
\E_{A} \left[\Pr_{k^* \sim \{0,1\}^\kappa} \left[\mathcal{A}^{\mathcal{O}[A]}\left(1^\kappa, \ket{\varphi_{k^*}} \right) = 1 \right] - \Pr_{\ket{\psi} \sim \mu_{\mathrm{Haar}}^n} \left[\mathcal{A}^{\mathcal{O}[A]}\left(1^\kappa, \ket{\psi} \right) = 1 \right] \right] \le \negl(\kappa).
\]
    By the same arguments as \Cref{cor:prs_abs_advantage_small} and \Cref{thm:single_copy_prs_prob_1}, this implies that with probability $1$ over $A$, $\{\ket{\varphi_k}\}_{k \in \{0,1\}^\kappa}$ is multi-copy pseudorandom relative to $\mathcal{O}[A]$.
\end{proof}

\ifauthors
\section*{Acknowledgments}
We thank Scott Aaronson, Chinmay Nirkhe, and Henry Yuen for insightful discussions. Part of this work was done while the authors attended the 2022 Extended Reunion for the Quantum Wave in Computing at the Simons Institute for the Theory of Computing.
\fi

\bibliographystyle{alphaurl}
\bibliography{bibliography}

\appendix

\section{Proof of \texorpdfstring{\Cref{thm:h2_h3}}{Theorem \ref{thm:h2_h3}}}
\label{app:proof_of_h2_h3}

We require a few lemmas that were stated and proved in \cite{AIK21-acrobatics}. First, we give a quantitative version of the BBBV Theorem \cite{BBBV97-search}.

\begin{lemma}[{\cite[Lemma 37]{AIK21-acrobatics}}]
\label{lem:average_case_bbbv}
Consider a quantum algorithm $Q^x$ that makes $T$ queries to $x \in \{\pm 1\}^N$. Let $y \in \{\pm 1\}^N$ be drawn from some distribution such that, for all $i \in [N]$, $\Pr_{y}\left[x_i \neq y_i\right] \le p$. Then for any $r > 0$:
\[
\Pr_{y} \left[ \left|\Pr\left[Q^y = 1 \right] - \Pr\left[Q^x = 1 \right] \right| \ge r \right] \le \frac{64pT^2}{r^2}.
\]
\end{lemma}

For $A: \{\pm 1\}^* \to \{\pm 1\}$ an oracle, let $B$ be defined depending on $A$ as in \Cref{def:PH_oracle}. For any $\ell \in \Naturals$, denote by $A_{\le \ell}$ (respectively, $B_{\le \ell}$) the concatenation of $A(x)$ (respectively, $B(x)$) over all strings $x$ of length at most $\ell$. The next lemma uses the standard connection between $\mathsf{PH}$ algorithms and $\mathsf{AC^0}$ circuits \cite{FSS84-circuit-oracle} to show that each bit of an oracle constructed as $B$ can be computed by a small $\mathsf{AC^0}$ circuit in the bits of $A$.

\begin{lemma}[{c.f.\ \cite[Lemma 35]{AIK21-acrobatics}}]
\label{lem:PH_to_AC0}
Fix $\ell, d \in \Naturals$, and let $\ell' \le \ell^{2^d}$. For each $x \in \{\pm 1\}^{\ell'}$, there exists an $\mathsf{AC^0}\left[O\left(2^{\ell^{2^d}}\right), 2d \right]$ circuit that takes as input $A_{\le \ell^{2^{d-1}}}$ and $B_{\le \ell}$ and computes $B(x)$.
\end{lemma}

Finally, we require the following concentration result for block sensitivity of $\mathsf{AC^0}$ circuits.

\begin{lemma}[{\cite[Lemma 45]{AIK21-acrobatics}}]
\label{lem:ac0_block_averaged_sensitivity_tail_bound}
Let $f: \{\pm 1\}^{MN} \to \{\pm 1\}$ be a circuit in $\mathsf{AC^0}[s, d]$. Let $x \in \{\pm 1\}^{MN}$ be an input, viewed as an $M \times N$ array with $M$ rows and $N$ columns. Let $y$ be sampled depending on $x$ as follows: uniformly select one of the rows of $x$, randomly reassign all of the bits of that row, and leave the other rows of $x$ unchanged. Then for any $p > 0$:
\[
\Pr_{x \sim \{\pm 1\}^{MN}}\left[\Pr_y\left[f(x) \neq f(y)\right] \ge p \right] \le 8M^2N \cdot 2^{-\Omega\left(\frac{p M}{(\log(s + MN))^d}\right)}.
\]
\end{lemma}

We extend \Cref{lem:ac0_block_averaged_sensitivity_tail_bound} to the case where $y$ is sampled by drawing one row of $x$ from $\mathcal{F}_n$. The proof is largely copied from \cite[Lemma 46]{AIK21-acrobatics}, but with different choices of parameters. In particular, we care more about the scaling in $M$ (which corresponds to the number of keys, $2^\kappa$) rather than $N$ (which corresponds to $2^{n + 1}$).

\begin{lemma}[{c.f. \cite[Lemma 46]{AIK21-acrobatics}}]
\label{lem:ac0_single_forrelated_block_indistinguishable_b}
Let $M \le N = 2^{n + 1} \le \quasipoly(M)$, and suppose that $f: \{\pm 1\}^{MN} \to \{\pm 1\}$ is a circuit in $\mathsf{AC^0}[\quasipoly(M), O(1)]$. Let $x \in \{\pm 1\}^{MN}$ be an input, viewed as an $M \times N$ array with $M$ rows and $N$ columns. Let $y$ be sampled depending on $x$ as follows: uniformly select one of the rows of $x$, randomly sample that row from $\mathcal{F}_n$\footnote{Viewing $\mathcal{F}_n$ as a distribution over $\{\pm 1\}^{2^{n + 1}}$ by converting functions into their truth tables.}, and leave the other rows of $x$ unchanged. Then for some $p = \frac{\polylog(M)}{\sqrt{M}}$, we have:
\[
\Pr_{x \sim \{\pm 1\}^{MN}}\left[\Pr_y\left[f(x) \neq f(y)\right] \ge p \right] \le 8M^2N \cdot 2^{-\Omega\left(\frac{\sqrt{M}}{\polylog(M)}\right)}.
\]
\end{lemma}

\begin{proof}
Consider a Boolean function $C(x, z, i)$ that takes inputs $x \in \{\pm 1\}^{MN}$, $z \in \{\pm 1\}^N$, and $i \in [M]$. Let $\tilde{y}$ be the string obtained from $x$ by replacing the $i$th row with $z$. Let $C$ output $1$ if $f(x) \neq f(\tilde{y})$, and $-1$ otherwise. Clearly, $C \in \mathsf{AC^0}[\quasipoly(M), O(1)]$. Observe that for any fixed $x$:
\begin{equation}
\label{eq:forrelation_block_same_distribution}
\Pr_{i \sim [M],z \sim \mathcal{F}_N}\left[C(x, z, i) = 1\right] = \Pr_{y}[f(x) \neq f(y)].
\end{equation}

By \Cref{thm:forrelation_AC0_pseudorandom}, there exists some $q \le \frac{\polylog(N)}{\sqrt{N}} \le \frac{\polylog(M)}{\sqrt{M}}$ such that:
\begin{equation}
\label{eq:forrelation_block_indistinguishable}
\left|\Pr_{i \sim [M],z \sim \mathcal{F}_N}\left[C(x, z, i) = 1\right] - \Pr_{i \sim [M],z \sim \{\pm 1\}^N}\left[C(x, z, i) = 1\right] \right| \le q.
\end{equation}

Choose $p = q + \frac{1}{\sqrt{M}}$, which is clearly at most $\frac{\polylog(M)}{\sqrt{M}}$. Putting these together, we obtain:
\begin{align*}
\Pr_{x \sim \{\pm 1\}^{MN}}\left[\Pr_y\left[f(x) \neq f(y)\right] \ge p \right]
&= \Pr_{x \sim \{\pm 1\}^{MN}}\left[\Pr_{i \sim [M], z \sim \mathcal{F}_N}\left[C(x, z, i) = 1\right] \ge p \right]\\
&\le \Pr_{x \sim \{\pm 1\}^{MN}}\left[\Pr_{i \sim [M], z \{\pm 1\}^N}\left[C(x, z, i) = 1\right] \ge p - q \right]\\
&= \Pr_{x \sim \{\pm 1\}^{MN}}\left[
\Pr_{i \sim [M],z \sim \{\pm 1\}^N}\left[f(x) \neq f(\tilde{y})\right]
\ge \frac{1}{\sqrt{M}}
\right]\\
&\le 8M^2N \cdot 2^{-\Omega\left(\frac{\sqrt{M}}{(\log(s + MN))^d}\right)}\\
&\le 8M^2N \cdot 2^{-\Omega\left(\frac{\sqrt{M}}{\polylog(M)}\right)},
\end{align*}
where the first line substitutes \eqref{eq:forrelation_block_same_distribution}; the second line holds by \eqref{eq:forrelation_block_indistinguishable} and the triangle inequality; the third line holds by the definition of $C$ and $\tilde{y}$ in terms of $i$ and $z$; the fourth line invokes \Cref{lem:ac0_block_averaged_sensitivity_tail_bound} for some $s = \quasipoly(M)$ and $d = O(1)$; and the last line uses these bounds on $s$ and $d$ along with the assumption that $N \le \quasipoly(M)$.
\end{proof}



We now proceed to the proof of \Cref{thm:h2_h3}, which we have restated here for convenience. The proof is conceptually similar to \cite[Theorem 47]{AIK21-acrobatics}

\thmappendix*

\begin{proof}
We will actually prove the stronger statement that the distinguishing advantage of $\mathcal{A}$ remains negligible between $\mathsf{H_2}$ and $\mathsf{H_3}$ even when $\mathcal{A}$ is given the entire truth table of $h$ as input, rather than $\ket{\Phi_h}$, so long as it only makes polynomially-bounded queries to $\mathcal{O}[A]$.

For $(A, h)$ sampled from $\mathsf{H_3}$, let $A^* \sim \mathcal{D}_{A, h}$ mean that we take $A^*$ to be an ``adjacent'' sample from $\mathsf{H_2}$. Formally, this means that we let $A^*$ be identical to $A$, except that for a uniformly random $k^* \in \{0,1\}^\kappa$ we replace $(f'_{k^*}, g'_{k^*}) \sim \mathcal{F}_n$. Note that this is well-defined because the functions $\{(f'_k, g'_k) \}_{k \in \{0,1\}^\kappa}$ are uniquely determined by $A$ and $h$ in the definition of $\mathsf{H_3}$. Additionally, observe that sampling $(A, h) \sim \mathsf{H_3}$ followed by $A^* \sim \mathcal{D}_{A, h}$ is equivalent to sampling $(A^*, h) \sim \mathsf{H_2}$. Hence, we may write:

\begin{align}
p_\mathsf{2}(\mathcal{A}) - p_\mathsf{3}(\mathcal{A}) &= \Pr_{(A, h) \sim \mathsf{H_2}} \left[
\mathcal{A}^{\mathcal{O}[A]}\left(1^\kappa, h \right) = 1 \right] - \Pr_{(A, h) \sim \mathsf{H_3}} \left[\mathcal{A}^{\mathcal{O}[A]}\left(1^\kappa, h \right) = 1 \right]\nonumber\\
&= \E_{(A, h) \sim \mathsf{H_3}} \left[
\Pr_{A^* \sim \mathcal{D}_{A, h}}\left[\mathcal{A}^{\mathcal{O}[A^*]}\left(1^\kappa, h \right) = 1 \right] - \Pr\left[\mathcal{A}^{\mathcal{O}[A]}\left(1^\kappa, h \right) = 1 \right] \right]\nonumber\\
&\le \E_{(A, h) \sim \mathsf{H_3}}\left|
\Pr_{A^* \sim \mathcal{D}_{A, h}} \left[\mathcal{A}^{\mathcal{O}[A^*]}\left(1^\kappa, h \right) = 1 \right] - \Pr\left[\mathcal{A}^{\mathcal{O}[A]}\left(1^\kappa, h \right) = 1 \right] \right|.\label{eq:p2_p3_A*}
\end{align}

Let $T \le \poly(\kappa)$ be an upper bound on the number of queries that $\mathcal{A}$ makes to $\mathcal{O}[A]$, and also on the length of the queries that it makes $\mathcal{O}[A]$. For some $p$ that we choose later, call $(A, h)$ ``good'' if, for all $x \in \{\pm 1\}^{\le T}$, $\Pr_{A^* \sim \mathcal{D}_{A, h}}\left[\mathcal{O}[A](x) \neq \mathcal{O}[A^*](x)\right] \le p$. Continuing from \eqref{eq:p2_p3_A*}, we have that for any $r \in [0, 1]$,

\begin{align}
p_\mathsf{2}(\mathcal{A}) - p_\mathsf{3}(\mathcal{A}) &\le \E_{(A, h) \sim \mathsf{H_3}}\left[ r + 
\Pr_{A^* \sim \mathcal{D}_{A, h}}\left[\left|\Pr\left[\mathcal{A}^{\mathcal{O}[A^*]}\left(1^\kappa, h \right) = 1 \right] - \Pr\left[\mathcal{A}^{\mathcal{O}[A]}\left(1^\kappa, h \right) = 1 \right]\right| \ge r\right] \right]\nonumber\\
&\le r + \frac{64pT^2}{r^2} + \Pr_{(A, h) \sim \mathsf{H_3}}[(A, h) \text{ not good}],\label{eq:p2_p3_Ah_not_good}
\end{align}
where in the second line we have applied \Cref{lem:average_case_bbbv}, which is applicable whenever $(A, h)$ is good.

It remains to upper bound the probability that $(A, h)$ is not good. By \Cref{lem:PH_to_AC0} with $\ell = \kappa$ and $\ell' = T$, for every $x \in \{\pm 1\}^{\le T}$, $\mathcal{O}[A](x)$ can be computed by an $\mathsf{AC^0}\left[O\left(2^{T^2} \right), O(1) \right]$ circuit whose inputs are $A_{\le T}$ and $B_{\le \kappa}$. Note that the only inputs to this circuit that can change between $A$ and $A^*$ are those corresponding to $\{(f_k, g_k)\}_{k \in \{0,1\}^\kappa}$; the rest can be viewed as fixed. We can further transform this into a circuit $C \in \mathsf{AC^0}\left[O\left(2^{T^2} \right), O(1) \right]$ whose inputs are $\{(f_k', g_k')\}_{k \in \{0,1\}^\kappa}$ and $h$, because $(f_k, g_k) = (f'_k, g'_k \cdot h)$ and each bit of this equivalence can be computed by a constant-size gadget. Using the notation of \Cref{lem:ac0_single_forrelated_block_indistinguishable_b} with $M = 2^\kappa$ and $N = 2^{n+1}$, we have that for some $p = \frac{\poly(\kappa)}{2^{\kappa/2}}$,
\begin{align*}
    \Pr_{(A, h) \sim \mathsf{H_3}}\left[\Pr_{A^* \sim \mathcal{D}_{A, h}}\left[\mathcal{O}[A](x) \neq \mathcal{O}[A^*](x)\right] > p\right] &= \Pr_{x \sim \{\pm 1\}^{MN}}\left[\Pr_{y}\left[C(x) \neq C(y)\right] > p\right]\\
    &\le 8M^2 N \cdot 2^{-\Omega\left(\frac{\sqrt{M}}{\polylog(M)} \right)}\\
    &\le 16 \cdot 2^{2 \kappa + n} \cdot 2^{-\Omega\left(\frac{2^{\kappa/2}}{\poly(\kappa)} \right)}\\
    &\le 2^{-2^{\Omega(\kappa)}}
\end{align*}
where the second line holds by \Cref{lem:ac0_single_forrelated_block_indistinguishable_b}, the third line substitutes $M$ and $N$, and the fourth line uses the assumption that $n \le \poly(\kappa)$. Hence, by a union bound over all $x \in \{\pm 1\}^{\le T}$, we conclude that
\begin{align*}
    \Pr_{(A, h) \sim \mathsf{H_3}}[(A, h) \text{ not good}] &= \Pr_{(A, h) \sim \mathsf{H_3}}\left[ \exists x \in \{\pm 1\}^{\le T} : \Pr_{A^* \sim \mathcal{D}_{A, h}}\left[\mathcal{O}[A](x) \neq \mathcal{O}[A^*](x)\right] > p\right]\\
    &\le \sum_{x \in \{\pm 1\}^{\le T}}\Pr_{(A, h) \sim \mathsf{H_3}}\left[ \Pr_{A^* \sim \mathcal{D}_{A, h}}\left[\mathcal{O}[A](x) \neq \mathcal{O}[A^*](x)\right] > p\right]\\
    &\le 2^{T + 1} \cdot 2^{-2^{\Omega(\kappa)}}\\
    &\le 2^{-2^{\Omega(\kappa)}},
\end{align*}
using in the third line the fact that $T \le \poly(\kappa)$. Combining with \eqref{eq:p2_p3_Ah_not_good} and choosing $r = 2^{-\kappa / 6}$ gives us the final bound
\begin{align*}
    p_\mathsf{2}(\mathcal{A}) - p_\mathsf{3}(\mathcal{A})
    &\le r + \frac{64pT^2}{r^2} + 2^{-2^{\Omega(\kappa)}}\\
    &\le r + \frac{\poly(\kappa)}{r^2 2^{\kappa/2}} + 2^{-2^{\Omega(\kappa)}}\\
    &\le \frac{\poly(\kappa)}{2^{\kappa / 6}}\\
    &\le \negl(\kappa),
\end{align*}
again because $T \le \poly(\kappa)$.
\end{proof}

\end{document}